\documentclass{article}
\usepackage{amsmath,amssymb,mathrsfs,amsthm}

\theoremstyle{plain}
\newtheorem{thm}{Theorem}
\newtheorem{prop}{Proposition}[section]
\newtheorem{corollary}{Corollary}[section]
\newtheorem{lemma}{Lemma}[section]

\theoremstyle{definition}

\newtheorem{assumption}{Assumption}

\theoremstyle{remark}
\newtheorem{remark}{Remark}[section]

\begin{document}

\title{Semiclassical study of shape resonances in the Stark effect}
\author{Kentaro Kameoka}

\date{}

\maketitle

\begin{abstract}
Semiclassical behavior of Stark resonances is studied. 
The complex distortion outside a cone is introduced to study  
resonances in any energy region for the Stark Hamiltonians with non-globally analytic potentials. 
The non-trapping resolvent estimate is proved by the escape function method. 
The Weyl law and the resonance expansion of the propagator are proved in the shape resonance model. 
To prove the resonance expansion theorem, 
the functional pseudodifferential calculus in the Stark effect is established, 
which is also useful in the study of the spectral shift function.
\end{abstract}

\section{Introduction}
In this paper, we study the semiclassical behavior of the resonances for the Stark Hamiltonian:
\[ P(\hbar)=-\hbar^2 \Delta +\beta x_1+V(x), \]
where $V(x)\in C^{\infty}(\mathbb{R}^n; \mathbb{R})$ is a non-globally analytic potential and $\beta>0$. 
Throughout this paper, the constant $\beta>0$ is fixed. 

We set the cone $C(K, \rho)=\{x \in \mathbb{R}^n | |x'| \le K(x_1+\rho)\}$, 
where $x'=(x_2, \dots, x_n)$, and denote its complement by $C(K, \rho)^c$. 
We denote the set of all bounded smooth functions with bounded derivatives by $C_b^{\infty}$.
Our assumption on the potential $V$ is as follows:

\begin{assumption}
The potential $V(x)\in C_b^{\infty}(\mathbb{R}^n; \mathbb{R})$ has an analytic continuation  
to the region $\{x\in \mathbb{C}^n|\mathrm{Re}x\in C(K_0, \rho_0)^c, |\mathrm{Im}x|<\delta_0\}$ 
for some $\rho_0\in \mathbb{R}, K_0>0$ and $\delta_0>0$, 
and $\partial V(x)$ goes to zero when $\mathrm{Re}x\to \infty$ in this region.
\end{assumption}

We introduce the complex distortion outside a cone to study semiclassical Stark resonances. 
This reduces the study of resonances to that of eigenvalues of a non-self-adjoint operator $P_{\theta}$. 
We take any $K>K_0$ and sufficiently large $\rho>0$ (such that Lemma 2.1 holds) and deform $P(\hbar)$ in $C(K, \rho)^c$. 
Take a convex set $\widetilde{C}(K, \rho)$ which has a smooth boundary 
such that $\widetilde{C}(K, \rho)$ is rotationally symmetric with respect to $x'$ and 
$\widetilde{C}(K, \rho)=C(K,\rho)$ in $x_1>-\rho+1$.
We define $F=-(1+K^{-2})^{\frac{1}{2}}\mathrm{dist}\left(\bullet, \widetilde{C}(K, \rho)\right)\ast \phi$, 
where $\phi \in C_c^{\infty}(\mathbb{R}^n)$, 
$\mathrm{supp}\phi \subset \{|x|<1\}$,  $\phi \ge 0$ and $\int\phi=1$. 
We also set $v(x)=(v_1(x), \dots, v_n(x))=\partial F(x)\in C_b^{\infty}(\mathbb{R}^n; \mathbb{R}^n)$. 
We next set $\Phi_{\theta}(x)=x+\theta v(x)$. 
This is a diffeomorphism for real $\theta$ with small $|\theta|$. 
We set $U_{\theta}f(x)=\left(\mathrm{det}\Phi_{\theta}'(x)\right)^{\frac{1}{2}}f(\Phi_{\theta}(x))$, 
which is unitary on $L^2(\mathbb{R}^n)$.
We define the distorted operator $P_{\theta}(\hbar)=U_{\theta}P(\hbar)U_{\theta}^{-1}$. 

The $P_{\theta}(\hbar)$ is an analytic family of closed operators for $\theta$ with 
$|\mathrm{Im}\theta|<\delta_0(1+K^{-2})^{-\frac{1}{2}}$ and $|\mathrm{Re}\theta|$ small
(Proposition 2.1). 
Moreover $P_{\theta}(\hbar)$ with $\mathrm{Im}\theta <0$ has discrete spectrum in 
$\{\mathrm{Im}z>\beta\mathrm{Im}\theta \}$ (Proposition 2.2). 
We note that we exclude the condition that $|\theta|$ is small by repeated applications of the Kato-Rellich theorem. 
We also note that we do not require that $\hbar$ is small. 

We set $L^p_{\mathrm{cone}}=\{f \in L^p | \mathrm{supp}f \subset C(K, \rho) \mspace{7mu} \text{for some} \mspace{7mu} K, \rho\}$ 
(in the following, we can replace $L^p_{\mathrm{cone}}$ by $L^p_{\mathrm{comp}}$). 
We also set $R_{+}(z, \hbar)=(z-P)^{-1}$ for $\mathrm{Im}z>0$.
Then we define the (outgoing) resonances of $P$ by meromorphic continuations of cutoff resolvents:
\begin{thm}
Suppose that Assumption 1 holds. 
Fix any $\hbar>0$. Then for any $\chi_1, \chi_2 \in L_{\mathrm{cone}}^{\infty}(\mathbb{R}^n)$ 
such that $\chi_j \not = 0$ on some open sets, the cutoff resolvent 
$\chi_1 R_{+}(z) \chi_2 \mspace{7mu} (\mathrm{Im}z>0)$ has a meromorphic continuation to 
$\mathrm{Im}z>-\beta \delta_0$ with finite rank poles. 
The pole $z$ is called a resonance and the multiplicity is defined by 
\[m_z=\mathrm{rank}\frac{1}{2\pi i}\oint_z \chi_1 R_{+}(\zeta) \chi_2 d\zeta .\]
The set of resonances is independent of the choices of $\chi_1$ and $\chi_2$ including multiplicities 
and denoted by $\mathrm{Res}(P)$. 
Moreover, $\mathrm{Res}(P)=\sigma_d(P_{\theta})$ including multiplicities 
in $\{\mathrm{Im}z>\beta\mathrm{Im}\theta\}$ if 
$0>\mathrm{Im}\theta >-\delta_0(1+K^{-2})^{-\frac{1}{2}}$ and 
$|\mathrm{Re}\theta|$ is small. 
\end{thm}

We emphasize that there is no restriction on $\mathrm{Re}z$ in Theorem 1. 
The resonances are also described including multiplicities in terms of meromorphic continuations 
of the matrix elements of the resolvent $(f, R_{+}(z)g)$ for $f, g \in L^{\infty}_{\mathrm{cone}}$ (Proposition 2.3) 
or $f, g\in \mathcal{A}=\{u\in L^2| \mathrm{supp}\hat{u} \mspace{7mu}\text{is compact}\}$ (Proposition 2.4). 
The latter formalism based on analytic vectors for $\frac{1}{i}\partial$ shows that our definition of resonances 
coincides with that based on the global analytic translation when the potential is globally analytic (Corollary 2.2).

The resonances for the Stark Hamiltonians have been investigated by many authors.
Avron-Herbst~\cite{AH} defined the Stark resonances by the translation analyticity. 
Herbst~\cite{H} defined the Stark resonances by the dilation analyticity. 
Herbst~\cite{H2} discussed the exponential decay of matrix elements of Stark propagator 
and its relation with Stark resonances. 

The resonance of $-\Delta +V(x)+\beta x_1$ near a negative eigenvalue $E$ of $-\Delta +V(x)$ 
and the exponentially small estimate of its width in the limit $\beta \to 0$ 
are studied by Sigal~\cite{Si} and Wang~\cite{W2} 
(see also Briet~\cite{Br} and Hislop-Sigal~\cite[Chapter 23]{HS}). 
These works employ the complex distortion in the half space.
Resonances for many body Stark Hamiltonians have been also studied 
(see Herbst-Simon~\cite{HeSi}, Sigal~\cite{Si2} and Wang~\cite{W3}). 

Dimassi-Petkov~\cite{DP} studied resonances of $-\hbar^2\Delta+V(x)+x_1$ and its relation with 
the spectral shift function in the semiclassical limit ($\hbar \to 0$). 
In \cite{DP}, resonances are defined and studied in the region $\mathrm{Re}z<R$ by the 
complex distortion in the region $x_1<R$. 

We next state the non-trapping resolvent estimate in our setting.
We denote the trapped set for the classical flow in the energy interval $[a,b]$ by $K_{[a,b]}$.
Thus $K_{[a,b]}$ is the set of all $(x_0, \xi_0)\in T^*\mathbb{R}^n$ 
such that $a\le p(x_0, \xi_0)\le b$ and $\sup_{t\in \mathbb{R}}|x(t)|<\infty$, 
where $(x(t), \xi(t))$ is the solution of the Hamilton equation for 
$p(x, \xi)=|\xi|^2+\beta x_1+V(x)$ with the initial value $(x_0, \xi_0)$.

Wang~\cite{W1} proved the non-trapping limiting absorbtion principle bound for the Stark Hamiltonians, 
that is, the $\mathcal{O}(\hbar^{-1})$ bound of $R_+(z, \hbar)$ for $\mathrm{Im}z>0$ with suitable weights 
(see also Hislop-Nakamura~\cite{HN}). 
The following bound implies the bound for the analytically continued cutoff resolvent 
$\chi R_+(z, h) \chi$ for $\mathrm{Im}z>-M\hbar\log\hbar^{-1}$, 
where $\chi \in L^{\infty}_{\mathrm{cone}}(\mathbb{R}^n)$, since 
$\chi R_+(z, h) \chi=\chi(z-P_{\theta}(\hbar))^{-1}\chi$ if $P_{\theta}$ is 
constructed by the deformation outside $\mathrm{supp}\chi$.
\begin{thm}
Suppose that Assumption 1 holds and $K_{[a,b]}=\emptyset$. 
Then for any $0<M\ll \widetilde{M}$ there exists $C>0$, 
which also depends on the construction of $P_{\theta}$, such that 
for small $\hbar>0$ and $z \in [a ,b]+i[-M\hbar\log\hbar^{-1}, \infty)$, 
\[\|(P_{\theta}(\hbar)-z)^{-1}\|\le C\exp(C(\mathrm{Im}z)_-/\hbar)/\hbar, \]
where $(\mathrm{Im}z)_-=\max\{-\mathrm{Im}z, 0\}$ 
and $\theta=-i\widetilde{M}\hbar\log\hbar^{-1}$. 
\end{thm}
The proof of Theorem 2 is based on the escape function method 
as in \cite{M}, \cite{SZ}, where the same result is proved for decaying potentials. 
Theorem 2 implies the non-trapping time decay estimate (Corollary 3.1) as in \cite{NSZ}.

Our principal motivation comes from the shape resonance model. 
Denote the full potential by $V_{\beta}=\beta x_1+V$. 

\begin{assumption}[shape resonance model]
Fix $a<b$. We assume that 
$\{x \in \mathbb{R}^n|V_{\beta}(x)\le b\}=\mathcal{G}^{\mathrm{int}}\cup\mathcal{G}^{\mathrm{ext}}$, 
where $\mathcal{G}^{\mathrm{int}}$ is compact and non-empty, 
$\mathcal{G}^{\mathrm{ext}}$ is closed, and $\mathcal{G}^{\mathrm{int}}\cap\mathcal{G}^{\mathrm{ext}}=\emptyset$. 
Moreover, we assume 
$K_{[a,b]}\cap\{(x,\xi)|x\in \mathcal{G}^{\mathrm{ext}}\}=\emptyset$.
\end{assumption}

Our first main theorem is the Weyl-type asymptotics for the Stark shape resonances: 
\begin{thm}
Under Assumption 1 and 2, there exists $S>0$ such that  
\[\lim_{\hbar \to 0}(2\pi \hbar)^n \#(\mathrm{Res}(P(\hbar))\cap ([a, b]-i[0, e^{-S/{\hbar}}]))
=\mathrm{Vol}(K_{[a,b]}).\]
\end{thm}
Our second main theorem is the resonance expansion theorem for Stark propagators 
(in this paper, the symbol $\mathcal{O}$ for some operator means $\mathcal{O}_{L^2 \to L^2}$ 
unless otherwise stated).
\begin{thm}
Suppose that Assumption 1 and 2 hold. 
Then for any $\psi \in C_c^{\infty}([a,b])$, $\delta>0$ 
and $\chi \in C_b^{\infty}(\mathbb{R}^n)\cap L^{\infty}_{\mathrm{cone}}(\mathbb{R}^n)$, 
there exist $a(\hbar)\in(a-\delta,a)$, $b(\hbar)\in (b,b+\delta)$ and $C>0$ such that for $t\ge C$, 
\begin{align*}
\chi e^{-itP/{\hbar}}\chi \psi(P)=&\sum_{z\in \mathrm{Res}(P(\hbar))\cap\Omega (\hbar)}
\mathrm{Res}_{\zeta=z}e^{-it\zeta/{\hbar}}\chi R_+(\zeta, \hbar)\chi\psi(P)+\mathcal{O}(\hbar^{\infty}),
\end{align*}
where $\Omega (\hbar)=[a(\hbar),b(\hbar)]-i[0, \hbar]$.
\end{thm} 
In the decaying potential case, Helffer-Sj\"{o}strand~\cite{HeSjo} and 
Stefanov~\cite{S} \cite{S2} proved Theorem 3.
Nakamura-Stefanov-Zworski~\cite{NSZ} provided a simplified proof of 
Theorem 3 and proved Theorem 4 after the work of Burq-Zworski~\cite{BZ}. 
We follow the general line of \cite{NSZ} with a minor simplification given by direct resolvent 
estimates (Proposition 4.1), which does not depend on the maximal principle technique 
(see Datchev-Vasy~\cite{DV1} \cite{DV2} for related resolvent estimates). 
Note that Theorem 4 is the resonance expansion in the limit $\hbar\to 0$ 
while the resonance expansion in Herbst~\cite{H2} is valid in the limit $t\to \infty$. 

To prove the resonance expansion theorem, we study the pseudodifferential property of $\psi(P)$. 
The symbol class is defined by 
\[S(m)=\{a(\bullet ; \hbar) \in C^{\infty}(T^*\mathbb{R}^n) | |\partial_{x, \xi}^{\alpha} a(x, \xi; \hbar)|
\le C_{\alpha} m(x, \xi) \}.\]
The Weyl quantization is defined by 
\[a^W(x, \hbar D; \hbar)u(x)=(2\pi \hbar)^{-n} 
\iint a(\frac{x+y}{2},\xi;\hbar)e^{i\langle x-y, \xi \rangle /{\hbar}}u(y)dyd\xi. \] 
We set $\sigma(x, \xi; y,\eta)=\langle \xi, y \rangle-\langle \eta, x \rangle$. 
The composition of Weyl symbols is 
\begin{align*}
(a\sharp b)(x, \xi)&=e^{\frac{i\hbar}{2}\sigma(D_x, D_{\xi}; D_y, D_{\eta})}a(x, \xi)b(y, \eta)|_{y=x, \eta=\xi}\\
&\sim \sum _{k\ge 0}\frac{1}{k!}\left(\frac{i\hbar}{2}\sigma(D_x, D_{\xi}; D_y, D_{\eta})\right)^k 
a(x, \xi)b(y, \eta)|_{y=x, \eta=\xi},
\end{align*}
which makes sense also for the formal power series. 
We denote $\mathrm{Op}S(m)=\{a^W(x, \hbar D; \hbar)|a \in S(m)\}$ and 
$S(m_1m_2^{-\infty})=\bigcap_{N>0}S(m_1m_2^{-N})$.

In the case where $\beta=0$, 
the usual functional pseudodifferential calculus implies 
$f(P) \in \mathrm{Op}S(\langle \xi \rangle^{-\infty})$ 
with the principal symbol $f(|\xi|^2+V(x))$ 
for $f\in C_c^{\infty}(\mathbb{R})$ (see~\cite[section 8]{DS}). 
In the case where $\beta>0$, 
this does not hold since $P$ is not elliptic in the semiclassical sense. 
In fact, $f(|\xi|^2+\beta x_1+V(x))\not \in S(m)$ for any tempered $m$ 
since $\partial_{\xi}^{\alpha}f(|\xi|^2+\beta x_1+V(x))$ 
involves the term $2^{|\alpha|}\xi^{\alpha}f^{(|\alpha|)}(|\xi|^2+\beta x_1+V(x))$ 
and $|\xi|$ can be arbitrary large on the support of $f(|\xi|^2+\beta x_1+V(x))$ 
when $x_1\to -\infty$. 
Thus $f(P)\not \in \mathrm{Op}S(m)$ for any tempered $m$.

Nevertheless, we can treat the weighted function $f(P)\chi$ and the difference of functions $f(P_2)-f(P_1)$.
We set $m=|\xi|^2+\langle x_1 \rangle$, where $ \langle x \rangle=(1+|x|^2)^{\frac{1}{2}}$.
Take $w \in  C^{\infty}(\mathbb{R}^n ; \mathbb{R}_{\ge 1})$ depending only on $x_1$ and 
$w=|x_1|$ for $x_1 \le -2$ and $w=1$ for $x_1 \ge -1$.

For the weighted function $f(P)\chi$, we prove the following. 
Suppose $V \in C_b^{\infty}(\mathbb{R}^n; \mathbb{R})$ and 
set $P(\hbar)=-\hbar^2 \Delta + \beta  x_1+V(x)$. 
\begin{thm}
Let $\chi \in S(w^{-\infty}\langle x' \rangle^{-s'})$ for some $s'\in \mathbb{R}$ 
and $f\in C_c^{\infty}(\mathbb{R})$. Then 
$f(P)\chi^W=a^W(x, \hbar D; \hbar)$ with $a\in S(m^{-\infty}\langle x' \rangle^{-s'})$ for $0<\hbar \le 1$. 
Moreover $a$ has an asymptotic expansion $a \sim \sum_{j=0}^{\infty} h^j a_j$ in $S(m^{-\infty}\langle x' \rangle^{-s'})$, 
which is the composition of the formal asymptotic expansion of the symbol of $f(P)$ and $\chi$. 
\end{thm}

We note that Theorem 5 holds true for $\chi^W f(P)$ since it is the adjoint of $\overline{f}(P)\overline{\chi}^W$.
\begin{remark}
In particular, $a_0=f(|\xi|^2+x_1+V(x))\chi(x, \xi)$ and 
$\mathrm{supp}a_j \subset \mathrm{supp}\chi\cap\left(\cup_{k \ge 1} \mathrm{supp}f^{(k)}(|\xi|^2+\beta x_1+V(x))\right)$ 
for $j \ge 1$. 
This implies that $(1-g)(P(\hbar))\chi^W f(P(\hbar))=\hbar^{\infty}\mathrm{Op}S(m^{-\infty})$ 
for $f, g\in C_c^{\infty}(\mathbb{R})$ with $g=1$ near $\mathrm{supp}f$. 
This is used in subsection 4.3. 
\end{remark}

For the difference of functions $f(P_2)-f(P_1)$, we prove the following.
Suppose $V_j \in C_b^{\infty}(\mathbb{R}^n; \mathbb{R})$ and 
set $P_j(\hbar)=-\hbar^2 \Delta + \beta  x_1+V_j(x)$, where $j=1,2$. 
 
\begin{thm}
Suppose $V_2-V_1 \in S(w^{-\infty}\langle x' \rangle^{-s'})$  for some $s'\in \mathbb{R}$ 
and let $f\in C_c^{\infty}(\mathbb{R})$. 
Then $f(P_2)-f(P_1)=a^W(x, \hbar D; \hbar)$ with $a\in S(m^{-\infty}\langle x' \rangle^{-s'})$ for $0<\hbar \le 1$. 
Moreover $a$ has an asymptotic expansion $a \sim \sum_{j=0}^{\infty} h^j a_j$ in $S(m^{-\infty}\langle x' \rangle^{-s'})$, 
which is the difference of the formal asymptotic expansion of the symbols of $f(P_2)$ and $f(P_1)$.
\end{thm}

\begin{corollary}
Suppose that the assumption in Theorem 6 holds with $s'>n-1$. 
Then the derivative of the spectral shift function $\xi'$ defined by 
$\langle \xi', f\rangle =\mathrm{tr}(f(P_2)-f(P_1))$ for $f\in C_c^{\infty}(\mathbb{R})$ has an asymptotic expansion 
$\xi'\sim (2\pi \hbar)^{-n}\sum_{j\ge 0}\hbar^j\tau_j$ in $\mathscr{D}'(\mathbb{R})$ (the space of distributions), where 
$\langle \tau_0, f \rangle =\iint(f(|\xi|^2+\beta x_1+V_2)- f(|\xi|^2+\beta x_1+V_1))dxd\xi$ and $\tau_1=0$.
\end{corollary}

We can also discuss the spectral shift function by the formula (\cite{RW}) 
$\mathrm{tr}(f(P)-f(P_0))=-\mathrm{tr}((\partial_{x_1}V)f(P))$ and Theorem 5, 
where $P_0=-\hbar^2\Delta+\beta x_1$. 
Dimassi-Petkov~\cite{DP} and Dimassi-Fujii\'{e}~\cite{DF} proved many properties of the spectral shift function by 
constructing an elliptic operator $\widetilde{P}$ such that 
$-\mathrm{tr}((\partial_{x_1}V)f(P))=-\mathrm{tr}((\partial_{x_1}V)f(\widetilde{P}))+\mathcal{O}(\hbar^{\infty})$.

\begin{remark}
The trace class property and finite terms in the asymptotic expansion can be discussed 
even if we only assume $V_1-V_2 \in S(w^{-M}\langle x' \rangle^{-s'})$ for large $M$ and $s'>n-1$.
\end{remark}

This paper is organized as follows.
In section 2, we define the Stark resonances in various manners and in particular prove Theorem 1.
In section 3, we prove the non-trapping resolvent estimate for the Stark Hamiltonian (Theorem 2). 
In section 4, we study the shape resonance model in the Stark effect 
and prove the Weyl-type asymptotics (Theorem 3) and the resonance expansion (Theorem 4). 
In section 5, we prove the functional pseudodifferential calculus in the Stark effect (Theorem 5, 6). 
In the Appendix, we justify the commutator calculations of the Stark resolvent in section 5.

\section{Definition of resonances}
Throughout this section, we assume Assumption 1. 
\subsection{Complex distortion}
We prove Theorem 1 in this subsection.
Recall from section 1 that $F=-(1+K^{-2})^{\frac{1}{2}}\mathrm{dist}\left(\bullet, \widetilde{C}(K, \rho)\right)\ast \phi$, 
$v(x)=(v_1(x), \dots, v_n(x))=\partial F(x)$, $\Phi_{\theta}(x)=x+\theta v(x)$, 
$U_{\theta}f(x)=\left(\mathrm{det}\Phi_{\theta}'(x)\right)^{\frac{1}{2}}f(\Phi_{\theta}(x))$, 
and $P_{\theta}(\hbar)=U_{\theta}P(\hbar)U_{\theta}^{-1}$. 
We first note that $F\in C^{\infty}(\mathbb{R}^n; \mathbb{R})$ is concave since $\widetilde{C}(K, \rho)$ is convex 
and the convolution with a positive function preserves convexity. 
We have $v_1(x)\ge 1$ on $C(K, \rho+1)^c$ by the coefficient $(1+K^{-2})^{\frac{1}{2}}$ in the 
definition of $F$. 
Moreover $(x_1)_- \partial^{\alpha}v_j$ is bounded for $|\alpha|\ge 1$. 
This follows from the replacement of $C(K, \rho)$ by $\widetilde{C}(K, \rho)$ for $|\alpha| = 1$ 
and from the mollification for $|\alpha|\ge 2$. 
We also note that $ \Phi_{\theta}'=I+\theta \partial^2 F$ is symmetric.
A calculation (using the invariance of Laplace-Beltrami operator) shows that
\begin{align*}
P_{\theta}(\hbar)
&=-\hbar^2 \sum_{i,j} g_{\theta}^{-\frac{1}{4}}\partial_i g_{\theta}^{\frac{1}{2}}g_{\theta}^{ij}
\partial_j g_{\theta}^{-\frac{1}{4}}+\beta x_1+\beta\theta v_1+V(\Phi_{\theta}(x))\\
&=-\hbar^2 \sum_{i,j} \partial_i g_{\theta}^{ij}\partial_j
+\hbar^2 r_{\theta}(x)+\beta x_1+\beta\theta v_1+V(\Phi_{\theta}(x)),
\end{align*}
where $(g_{\theta}^{ij})=(\Phi_{\theta}')^{-2}$, $g_{\theta}=\mathrm{det}(\Phi_{\theta}')^{2}$ 
and $r_{\theta}=-\sum_{i,j}g_{\theta}^{-\frac{1}{4}}
(\partial_i(g_{\theta}^{\frac{1}{2}}g_{\theta}^{ij}\partial_j g_{\theta}^{-\frac{1}{4}}))$.
This expression defines $P_{\theta}(\hbar)$ as a differential operator 
for complex $\theta$ with small $|\mathrm{Re}\theta|$ and $|\mathrm{Im}\theta|<(1+K^{-2})^{-\frac{1}{2}}\delta_0$. 
We denote the semiclassical principal symbol of $P_{\theta}(\hbar)$ by 
\[p_{\theta}=\langle (I+\theta F'')^{-1}\xi, (I+\theta F'')^{-1}\xi \rangle
+\beta x_1+\beta \theta v_1+V(\Phi_{\theta}(x)).\]
An advantage of our definition of $P_{\theta}(\hbar)$ is as follows:
\begin{lemma}
For $\mathrm{Im}\theta \le 0$, $\mathrm{Im}(-\hbar^2 \sum_{i,j} \partial_i g_{\theta}^{ij}\partial_j)\le 0$ 
in the form sense. 
If $\rho>0$ is large and $\mathrm{Im}\theta \le 0$, then $\mathrm{Im}p_{\theta}\le 
-\frac{1}{2}\beta|\mathrm{Im}\theta|v_1(x) \le 0$ on $T^*\mathbb{R}^n$.
\end{lemma}
\begin{proof}
Since $F$ is concave, $\mathrm{Im}(\langle (I+\theta F'')^{-1}\xi, (I+\theta F'')^{-1}\xi \rangle)\le 0$ 
by diagonalizing $F''$. 
This also implies the first statement. 
We have $|\mathrm{Im}V(\Phi_{\theta}(x))|\lesssim |\mathrm{Im}\theta|\sup|\partial V(y) \cdot v(x)|$, 
where $y$ ranges over a small complex neighborhood of $x$. 
Thus for large $\rho$, $|\mathrm{Im}V(\Phi_{\theta}(x))|\le \varepsilon |\mathrm{Im}\theta| ||v(x)|, \varepsilon \ll 1$. 
Since $v_1(x)\ge c|v(x)|$, we have 
$\mathrm{Im}(\beta\theta v_1+V(\Phi_{\theta}(x)))\le -\frac{1}{2}\beta|\mathrm{Im}\theta|v_1(x) \le 0$.
\end{proof}

We next study the operator-theoretic property of $P_{\theta}$. 
Since $(P_{\theta}u_1, u_2)=(u_1, P_{\overline{\theta}}u_2)$ for $u_1, u_2\in C_c^{\infty}$, $P_{\theta}(\hbar)$ 
is closable on $C_c^{\infty}$ and the closure is also denoted by $P_{\theta}(\hbar)$. 
We first prove the analyticity of $P_{\theta}$ with respect to $\theta$. 
\begin{prop} 
For $0<\hbar\le1$, $P_{\theta}$ is an analytic family of type (A) with respect to $\theta$ with 
$|\mathrm{Im}\theta|<\delta_0(1+K^{-2})^{-\frac{1}{2}}$ and $|\mathrm{Re}\theta|$ small.
That is, $D(P_{\theta})=D(P)$ and $P_{\theta}u$ is analytic with respect to $\theta$ for any $u \in D(P)=D(P_{\theta})$. 
Thus, $(P_{\theta}-z)^{-1}$ is analytic with respect to $\theta$. 
Moreover, $P_{\theta}^*=P_{\overline{\theta}}$.
\end{prop}
\begin{proof}
We prove $\|(P_{\theta}-P_{\theta'})u\|\le C |\theta-\theta'||1+\theta|^2 \|P_{\theta}u\|+C_{\theta}\|u\|$ 
for $u \in C_c^{\infty}$, 
where $C$ is independent of $\theta$ with $|\mathrm{Re}\theta|$ small.
We only have to estimate 
$\|(\hbar^2 \sum_{i,j} \partial_i g_{\theta}^{ij}\partial_j-\hbar^2 \sum_{i,j} \partial_i g_{\theta'}^{ij}\partial_j)u \|$. 
Take $w \in  C^{\infty}(\mathbb{R}^n ;\mathbb{R}_{\ge 1})$ depending only on $x_1$ and 
$w=|x_1|$ for $x_1 \le -2$ and $w=1$ for $x_1 \ge -1$. 
Since $(x_1)_- \partial^{\alpha}v_j$ is bounded for $|\alpha|\ge 1$ and 
$\mathrm{Re}\sum g_{\theta}^{ij}\xi_i\xi_j \ge c|1+\theta|^{-2}|\xi|^2$ for small $|\mathrm{Re}\theta|$, 
\begin{align*}
&\|(\hbar^2 \sum_{i,j} \partial_i g_{\theta}^{ij}\partial_j-\hbar^2 \sum_{i,j} \partial_i g_{\theta'}^{ij}\partial_j)u \| \\
&\le C |\theta-\theta'|\|w^{-1}u\|_{H^2_{\hbar} }\\
&\le C  |\theta-\theta'||1+\theta|^2\|w^{-1}\hbar^2 \sum_{i,j} \partial_i g_{\theta}^{ij}\partial_ju\|
+C|\theta-\theta'|\|w^{-1}u\|\\
&\le  C|\theta-\theta'||1+\theta|^2\|x_1w^{-1}u\|+C|\theta-\theta'||1+\theta|^2\|w^{-1}P_{\theta}u\|+ C_{\theta}\|u\|. 
\end{align*}
The first term can be estimated as follows. 
We take $\chi (x_1)$ such that $\chi  (x_1)=0$ for $x_1\le1$ and $\chi(x_1)=1$ for $x_1 \ge2$. Then 
$\|x_1w^{-1}u\|\le C\|x_1 \chi u\|+C\|u\|\le C\|P_{\theta}\chi u\|+C\|u\|
\le C \|[P_{\theta}, \chi]u\|+C\|P_{\theta}u\|+C\|u\|
\le C\|P_{\theta}u\|+C_{\theta}\|u\|$, 
where the last inequality follows from the standard elliptic estimate.

Repeated applications of Kato-Rellich theorem (see \cite[section X.2]{RS}) to 
$\begin{pmatrix} 0 & P_{\overline{\theta}} \\ P_{\theta} & 0 \end{pmatrix}$
show that $P_{\theta}$ is closed on $D(P_{\theta})=D(P)$ 
and $P_{\overline{\theta}} = P_{\theta}^*$. 
This is valid for small $|\mathrm{Re}\theta|$ and $|\mathrm{Im}\theta|<(1+K^{-2})^{-\frac{1}{2}}\delta_0$ 
since $\lim_{k\to \infty}a_k=\infty$ if $a_{k+1}=a_k+\frac{c}{(1+a_k)^2}$ and $a_0=0$.

Since $P_{\theta}u$ is analytic with respect to $\theta$ for $u\in C_c^{\infty}$, 
an approximation argument shows that $P_{\theta}u$ is analytic with respect to $\theta$ for $u \in D(P)$. 
This implies that $(P_{\theta}-z)^{-1}$ is analytic with respect to $\theta$ by the general theory 
(see~\cite[section 7.1, section 7.2]{K}). 
\end{proof}
We next prove the discreteness of the spectrum of $P_{\theta}$ in $\{\mathrm{Im}z>\beta\mathrm{Im}\theta\}$. 
\begin{prop}
Fix $\theta$ with $-\delta_0(1+K^{-2})^{-\frac{1}{2}}<\mathrm{Im}\theta<0$ and $|\mathrm{Re}\theta|$ small. 
Then for $0<\hbar\le 1$, 
$P_{\theta}-z$ is an analytic family of Fredholm operators with index 0 on $\{\mathrm{Im}z>\beta\mathrm{Im}\theta\}$ 
and invertible for $\mathrm{Im}z \gg 1$. 
Thus $(P_{\theta}-z)^{-1}$ is meromorphic on $\{\mathrm{Im}z>\beta\mathrm{Im}\theta\}$ 
with finite rank poles.
\end{prop}
\begin{remark}
In fact, $P_{\theta}-z$ is invertible for $\mathrm{Im}z \ge 0$ by Theorem 1, Corollary 2.1 and Remark 2.4.
\end{remark}
\begin{proof}
Set $\widetilde{P_{\theta}}=P_{\theta}-iM\phi(x/M)\phi(\hbar D/M)^2\phi(x/M)$, 
where $M>1$, $0\le \phi \in C_c^{\infty}(\mathbb{R}^n)$, $\phi=1$ near $\{|x|\le 1/3\}$, 
$\mathrm{supp}\phi \subset \{|x|\le 1\}$ and $\int \phi =1$. 
Take $\Omega \Subset \{\mathrm{Im}z>\beta\mathrm{Im}\theta\}$. 
We prove that $\|(\widetilde{P_{\theta}}-z)^{-1}\|\le C$ for $0<\hbar\le 1$ and $z\in \Omega$ for large $M>1$.

Take $1\ll R \ll M$ and let $\chi_1, \chi_2\in C_b^{\infty}(\mathbb{R}^n)$ be 
cutoff functions near $C(K,R)$ and $C(K,R)^c$ respectively. 
We first note that $-\mathrm{Im}(\chi_2 u, (\widetilde{P_{\theta}}-z)\chi_2 u)
\ge c\|\chi _2 u\|^2 -\mathcal{O}(R^{-1})\|u\|^2$ since 
$\mathrm{Im}(\beta\theta v_1+V(\Phi_{\theta}(x))-z)\le -c$ near $C(K,R)^c$ by Lemma 2.1 
and $r_{\theta}(x)=\mathcal{O}(R^{-1})$ near $C(K,R)^c$. 
Thus we can take large $R>0$ such that $\|(\widetilde{P_{\theta}}-z)\chi_2 u\|\ge c\|\chi _2 u\|$. 

We next prove $\|(\widetilde{P_{\theta}}-z)\chi_1 u \|\ge c\|\chi _1u\|$ for large $M>R$. 
We take small $\varepsilon>0$ and set $\chi_{j,M}=\tau_j(G(x)/M)$, where 
$\tau_1 \in C_b^{\infty}(\mathbb{R})$ is a cutoff near 
$(-\infty, \varepsilon]$, $\tau_2 \in C_b^{\infty}(\mathbb{R})$ is a cutoff near $[2\varepsilon, \infty)$ and 
$G(x)=(1+K^{-2})^{\frac{1}{2}}\mathrm{dist}\left(\bullet,  \widetilde{C}(K, R)\right)\ast \phi$, 
where $\phi$ is as above. Then $\chi_{1,M},\chi_{2,M}\in C_b^{\infty}$, 
$\|\partial^{\alpha} \chi_{j,M}\|_{\infty}=\mathcal{O}(M^{-1})$ for $|\alpha|\ge 1$, 
$\chi_{1,M}=1$ near $\mathrm{supp}\partial \chi_j$, 
$\chi_{2,M}=1$ on $C(K,R+2\varepsilon M)^c$, $\chi_{2,M}=0$ on $\mathrm{supp}\chi_1$ and
$\mathrm{supp}\chi_{1,M}\cap\mathrm{supp}\chi_{2,M}=\emptyset$.
Take $w \in  C^{\infty}(\mathbb{R}^n; \mathbb{R}_{\ge 1})$ depending only on $x_1$ 
and $w=|x_1|$ for $x_1 \le -2$ and $w=1$ for $x_1 \ge -1$.
We set $Q=\widetilde{P_{\theta}}-z+\beta\chi_{2,M}w-iM\chi_{2,M}$. 
We now prove that $Q^{-1}:H_{\hbar}^k \to H_{\hbar}^{k+2}$ is uniformly bounded with respect to large $M>1$ 
for any $k$, where $H_{\hbar}^k=\langle \hbar D \rangle^{-k}L^2$.

Denote the seminorms in $S(\langle \xi \rangle^{N})$ by 
$|a|_{N,\alpha}=\sup_{x,\xi}|\partial_{x,\xi}^{\alpha}a|/\langle \xi \rangle^{N}$. 
We set $Q=q^W$. 
Then $q=\sum  g_{\theta}^{ij}\xi_i\xi_j+\beta x_1-iM\phi(x/M)^2\phi(\xi/M)^2+\beta\chi_{2,M}w-iM\chi_{2,M}+k_M(x, \xi)$, 
where $k_M$ is bounded in $S(1)$ with respect to $M>1$.
We note that for $|\alpha|\ge1$, $\sup_{M>1}|M\phi(x/M)^2\phi(\xi/M)^2|_{0,\alpha}<\infty$, 
$\sup_{M>1}|\chi_{2,M}w|_{0,\alpha}<\infty$ and $\sup_{M>1}|iM\chi_{2,M}|_{0,\alpha}<\infty$ since 
$\mathrm{supp}\partial \chi_{2,M}\subset \{x_1>-CM\}$ and 
$\|\partial^{\alpha} \chi_{2,M}\|_{\infty}=\mathcal{O}(M^{-1})$ for $|\alpha|\ge 1$. 
We also recall that $\mathrm{Re}\sum  g_{\theta}^{ij}\xi_i\xi_j\ge c|\xi|^2$ for some $c>0$ and 
$\mathrm{Im}\sum g_{\theta}^{ij}\xi_i\xi_j\le 0$.
Thus $|q^{-1}|_{-2,\alpha}\le C\sup_{x,\xi}B_\kappa(x,\xi)$ for $|\alpha|=\kappa$ if we set   
\[B_\kappa=\langle \xi \rangle^{\kappa+2}/|c|\xi|^2+\beta x_1-iM\phi(x/M)^2\phi(\xi/M)^2+\chi_{2,M}\beta w-iM\chi_{2,M}+k_M|^{\kappa+1}.\]

We have $\mathbb{R}^n=\{|x|<M/3\}\cup C(K,R+2\varepsilon M)^c\cup \{x_1>cM\}$ for some $c>0$ since $\varepsilon$ is small. 
Take large $C_1>0$. 
For $|x|<M/3, |\xi|<C_1M^{1/2}$, we see $B_\kappa\le CM^{(\kappa+2)/2}/M^{\kappa+1}=CM^{-\kappa/2}$ in view of $iM\phi(x/M)^2\phi(\xi/M)^2$.    
For $|x|<M/3, |\xi|>C_1M^{1/2}$, we see $B_\kappa\le C|\xi|^{\kappa+2}/(c|\xi|^2-\beta M+k_M)^{\kappa+1}\le C|\xi|^{\kappa+2}/|\xi|^{2\kappa+2}
=C|\xi|^{-\kappa}\le CM^{-\kappa/2}$ since $c|\xi|^2\gg\beta M$ by $C_1\gg 1$. 
For $x \in C(K,R+2\varepsilon M)^c$, we see 
$B_\kappa\le C\langle \xi \rangle^{\kappa+2}/|c|\xi|^2-iM+k_M|^{\kappa+1}$ in view of $\chi_{2,M}\beta w-iM\chi_{2,M}$. 
This is bounded by $CM^{-\kappa/2}$ by considering $|\xi|\lessgtr C_1M^{1/2}$.
For $x_1>cM$, we see $B_\kappa(x,\xi)\le C\langle \xi\rangle^{\kappa+2}/(|\xi|^2+M+k_M)^{\kappa+1}\le CM^{-\kappa/2}$ 
by considering $|\xi|\lessgtr C_1M^{1/2}$. 

Thus we have proved $|q^{-1}|_{-2,\alpha}=\mathcal{O}(M^{-|\alpha|/2})$.
Thus we see that $(q^{-1})^W :H_{\hbar}^k \to H_{\hbar}^{k+2}$ is uniformly bounded with respect to $M>1$. 
We also see that $\lim_{M\to \infty}q_1=0$ in $S(1)$ if $q^{-1}\sharp q=1+q_1$ 
since $\partial_{x,\xi} q$ is bounded in $S(\langle \xi \rangle^2)$ with respect to $M$ and
$\lim_{M\to \infty}\partial_{x,\xi} q^{-1}=0$ in $S(\langle \xi \rangle^{-2})$. 
Thus $(1+q_1^W)^{-1}:H_{\hbar}^k \to H_{\hbar}^{k}$ is uniformly bounded with respect to large $M>1$. Thus 
$Q^{-1}:H_{\hbar}^k \to H_{\hbar}^{k+2}$ is uniformly bounded with respect to large $M>1$ (in fact 
$Q^{-1}\in \mathrm{Op}S(\langle \xi \rangle ^{-2})$ uniformly for large $M$ by Beals's theorem).
Thus $\|\chi _1u\|=\|Q^{-1}Q\chi _1u\|\le C\|Q\chi _1u\|=C\|(\widetilde{P_{\theta}}-z)\chi _1u\|$ 
since $\chi_{2,M}=0$ on $\mathrm{supp}\chi_1$. 
Thus we have 
\[\|u\|\le \sum \|\chi_j u\| \le C\sum \|(\widetilde{P_{\theta}}-z)\chi_j u\|
\le C\|(\widetilde{P_{\theta}}-z)u\|+C\sum \|[\widetilde{P_{\theta}}, \chi_j]u\|.\]

We finally estimate $\|[\widetilde{P_{\theta}}, \chi_j]u\|$.
Since $\chi_{1,M}=1$ near $\mathrm{supp}\partial\chi_j$ and 
$\partial_{\xi}(\widetilde{p_{\theta}}-z)$ is bounded in $S(\langle \xi \rangle)$ with respect to $M>1$, 
we have 
\begin{align*}
\|[\widetilde{P_{\theta}}, \chi_j]u\|&
\le \|[\widetilde{P_{\theta}}, \chi_j]\chi_{1,M}u\|+\|[\widetilde{P_{\theta}}, \chi_j](1-\chi_{1,M})u\| \\
&\le C\|\chi_{1,M}u\|_{H^1_{\hbar}}+\mathcal{O}(M^{-\infty})\|u\|_{L^2}.
\end{align*} 
Since $Q^{-1}:H_{\hbar}^{-1} \to H_{\hbar}^{1}$ is uniformly bounded with respect to large $M>1$ we have 
$ \|\chi_{1,M} u\|_{H^1_{\hbar}}\le C\|Q\chi_{1,M} u\|_{H^{-1}_{\hbar}}$. 
Since $\mathrm{supp}\chi_{1,M}\cap\mathrm{supp}\chi_{2,M}=\emptyset$, we have 
$\|Q\chi_{1,M} u\|_{H^{-1}_{\hbar}}=\|(\widetilde{P_{\theta}}-z)\chi_{1,M}u\|_{H_{\hbar}^{-1}} 
\le \|(\widetilde{P_{\theta}}-z)u\|_{L^2}+\|[\widetilde{P_{\theta}},\chi_{1,M}] u\|_{H_{\hbar}^{-1}}$.
Since $\partial_{\xi}(\widetilde{p_{\theta}}-z)$ is bounded in $S(\langle \xi \rangle)$ with respect to $M>1$ and 
$\partial\chi_{1,M} =\mathcal{O}(M^{-1})$ in $S(1)$, 
we have $\|[\widetilde{P_{\theta}},\chi_{1,M}] u\|_{H_{\hbar}^{-1}}\le CM^{-1}\|u\|_{L^2}$.
Thus we have $\| (\widetilde{P_{\theta}}-z)u\|\ge c \|u\|$ for large $M>1$ and $0<\hbar\le 1$. 
 
We also have $\| (\widetilde{P_{\theta}}-z)^*u\|\ge c \|u\|$ for large $M>1$ since 
$(\widetilde{P_{\theta}}-z)^*=P_{\overline{\theta}}+iM\phi(x/M)\phi(\hbar D/M)^2\phi(x/M)-\overline{z}$ 
by Proposition 2.1. 
Banach's closed range theorem thus implies that $\widetilde{P_{\theta}}-z$ is invertible and 
$\|(\widetilde{P_{\theta}}-z)^{-1}\|\le C$ for $0<\hbar\le 1$ and $z\in \Omega$ for large $M>1$.
Since $M\phi(x/M)\phi(\hbar D/M)^2\phi(x/M)$ is compact, 
$P_{\theta}-z=(1+iM\phi(x/M)\phi(\hbar D/M)^2\phi(x/M)(\widetilde{P_{\theta}}-z)^{-1})(\widetilde{P_{\theta}}-z)$ 
is Fredholm with index 0. 
Finally, $P_{\theta}-z_0$ is invertible for $\mathrm{Im}z_0 \gg 1$ 
since $-\mathrm{Im}(u, (P_{\theta}-z)u)\ge \mathrm{Im}z_0\| u\|^2 -C\hbar^2\|u\|^2$ by Lemma 2.1.
\end{proof}
\begin{remark}
The proof will be simplified if we assume that $0<\hbar\ll 1$.
\end{remark}
\begin{proof}[Proof of Theorem 1]
Take any $0<\delta_1<\delta_0$. 
Take $\chi_1, \chi_2 \in L_{\mathrm{cone}}^{\infty}(\mathbb{R}^n)$ such that $\chi_j \not = 0$ on some open sets. 
Construct $P_{\theta}$ outside $\mathrm{supp}\chi_j$ and $C(K, \rho)$ with $(1+K^{-2})^{-\frac{1}{2}}\delta_0>\delta_1$. 
Then $\chi_1 R_{+}(z) \chi_2=\chi_1 U_{\theta}R_{+}(z)U_{\theta}^{-1} \chi_2
=\chi_1 (z-P_{\theta})^{-1} \chi_2$ for real $\theta$ and $\mathrm{Im}z>0$. 
The right hand side has an analytic continuation with respect to $\theta$ with $|\mathrm{Im}\theta|<\delta_1$ 
and $|\mathrm{Re}\theta|$ small by Proposition 2.1.
This in turn implies that the left hand side has a meromorphic continuation to 
$\mathrm{Im}z>-\beta \delta_1$ by Proposition 2.2. 
If $z\not\in \sigma_d(P_{\theta})$, this is analytic near $z$. 
Suppose that $z\in \sigma_d(P_{\theta})$.
Then the multiplicity of the pole $z$ of $\chi_1 R_{+}(z) \chi_2$ is 
given by $\mathrm{rank}\frac{1}{2\pi i}\oint_z \chi_1 R_{+}(\zeta)\chi_2d\zeta
=\mathrm{rank}\frac{1}{2\pi i}\oint_z \chi_1 (\zeta-P_{\theta})^{-1}\chi_2d\zeta
=\mathrm{rank}\chi_1 \Pi_z^{\theta} \chi_2$, where 
$\Pi_z^{\theta}=\frac{1}{2\pi i}\oint_z  (\zeta-P_{\theta})^{-1}d\zeta$ is 
the generalized eigenprojection of $P_{\theta}$ at $z$. 
We have $ (P_{\theta}-z)^k \Pi_z^{\theta}=0$ for some $k$ by the general theory of closed operators. 
Then the repeated applications of the unique continuation theorem for second order elliptic operators imply that 
$\mathrm{rank}\chi_1 \Pi_z^{\theta}=\mathrm{rank}\Pi_z^{\theta}$. 
Since $(\Pi_z^{\theta})^*=\Pi_{\overline{z}}^{\overline{\theta}}$, the same argument for the adjoint implies that 
$\mathrm{rank}\chi_1 \Pi_z^{\theta}=\mathrm{rank}\chi_1 \Pi_z^{\theta}\chi_2$. 
This proves that the definition of resonances is independent of $\chi_1$, $\chi_2$ 
and the multiplicity is given by $m_z=\mathrm{rank}\Pi_z^{\theta}$.
\end{proof}

\begin{remark}
The facts that $\|(\widetilde{P_{\theta}}-z)^{-1}\|=\mathcal{O}(1)$ for $z \in \Omega$ and 
$\|(P_{\theta}-z_0)^{-1}\|=\mathcal{O}(1)$ if $\mathrm{Im}z_0>0$ in the proof of Proposition 2.2 imply 
the following general upper bound on the number of the resonances; 
if $\Omega \Subset \{\mathrm{Im}z>-\beta\delta_0\}$, then 
\[\#(\mathrm{Res}(P(\hbar))\cap \Omega)=\mathcal{O}(\hbar^{-n})\] 
and the following a priori resolvent bound; 
if $z\in\Omega \Subset \{\mathrm{Im}z>\beta\mathrm{Im}\theta\}$, $0<\delta(\hbar)<c<1$ and 
$\mathrm{dist}(z, \mathrm{Res}(P(\hbar)))\ge \delta(\hbar)$, then  
\[ \|(P_{\theta}-z)^{-1}\|\le C \exp(C\hbar^{-n}\log\frac{1}{\delta(\hbar)}).\]
See \cite[section 7.2]{DZ} for the proof.
\end{remark}

\subsection{Meromorphic continuations of matrix elements} 
The resonances are also described by meromorphic continuations of the matrix elements of the resolvent.
\begin{prop}
The matrix element of the resolvent $(f, R_{+}(z)g)$ has a meromorphic continuation to 
$ \mathrm{Im}z>-\beta \delta_0$ for any $f, g \in L_{\mathrm{cone}}^{2}$. 
For $z$ with $ \mathrm{Im}z>-\beta \delta_0$, $z$ is a resonance of $P$ 
if and only if $z$ is a pole of $(f, R_{+}(z)g)$ for some $f, g \in  L_{\mathrm{cone}}^{2}$ and 
the multiplicity $m_z$ is given by the maximal number $k$ such that
there exist $ f_1, \dots ,f_k , g_1, \dots, g_k \in L_{\mathrm{cone}}^{2}$ 
with  $\det(\frac{1}{2\pi i}\oint_z(f_i, R_{+}(\zeta)g_j)d\zeta)_{i,j=1}^{k} \not =0$. 

Moreover, for any nonempty open bounded $U \subset \mathbb{R}^n$ 
and an orthonormal basis $\{f_i\}$ of $L^2(U)$, 
$m_z=\mathrm{rank}(\frac{1}{2\pi i}\oint_z(f_i, R_{+}(\zeta)f_j)d\zeta)_{i,j=1}^{\infty}$.
\end{prop}

\begin{proof}
Take $\chi_1, \chi_2$ as in Theorem 1 and 
set $\Pi_z^{\chi_1, \chi_2}=\frac{1}{2\pi i}\oint_z \chi_1 R_{+}(\zeta) \chi_2 d\zeta$. 
Then $m_z=\mathrm{rank}\Pi_z^{\chi_1, \chi_2}$. 
We have  
$(f,\Pi_z^{\chi_1, \chi_2}g)=(f, \frac{1}{2\pi i}\oint_z \chi_1 R_{+}(\zeta) \chi_2 d\zeta g)
=\frac{1}{2\pi i}\oint_z(\overline{\chi_1}f, R_{+}(\zeta)\chi_2g)d\zeta$. 
The Proposition easily follows from this.
\end{proof}

\begin{corollary}
$\mathrm{Res}(P)\cap \mathbb{R}=\sigma_{pp}(P)$.
\end{corollary}

\begin{proof}
This follows from Proposition 2.3 
and the formula $\lim_{\varepsilon \to +0}\varepsilon(f, (P-\lambda-i\varepsilon)^{-1}g)=i(f, E_{\{\lambda\}}g)$.
\end{proof}

\begin{remark}
The absence of embedded eigenvalues $\sigma_{pp}(P)=\emptyset$ for the Stark Hamiltonian was proved by Avron-Herbst~\cite{AH}.
\end{remark}

The resonances are also described based on analytic vectors. 
Set $\mathcal{A}=\{u\in L^2(\mathbb{R}^n)| \mathrm{supp}\hat{u} \mspace{7mu}\text{is compact}\}$, 
which consists of analytic vectors for the generators of the translations 
$(\frac{1}{i}\partial_1, \dots , \frac{1}{i}\partial_n)$. 
\begin{prop}
The matrix element of the resolvent $(f, R_{+}(z)g)$ has a meromorphic continuation to 
$ \mathrm{Im}z>-\beta \delta_0$ for any $f, g \in \mathcal{A}$. 
For $z$ with $ \mathrm{Im}z>-\beta \delta_0$, $z$ is a resonance of $P$ 
if and only if $z$ is a pole of $(f, R_{+}(z)g)$ for some $f, g \in \mathcal{A}$ 
and the multiplicity is given by the maximal number $k$ such that 
there exist $ f_1, \dots ,f_k , g_1, \dots, g_k \in \mathcal{A} $
with $\det(\frac{1}{2\pi i}\oint_z(f_i, R_{+}(\zeta)g_j)d\zeta)_{i,j=1}^{k} \not =0 $.
\end{prop}

\begin{proof}
Take any $0<\delta_1<\delta_0$ and construct $P_{\theta}$ 
outside $C(K, \rho)$ with $(1+K^{-2})^{-\frac{1}{2}}\delta_0>\delta_1$. 
We first note that $U_{\theta}f$ ($f\in \mathcal{A}$) has an analytic continuation for small 
$|\mathrm{Re}\theta|$ by the definition of $\mathcal{A}$. 
Take $f, g \in \mathcal{A}$. 
Then $(f, R_{+}(z)g)= (U_{\theta}f, U_{\theta}R_{+}(z)U_{\theta}^{-1}U_{\theta}g)
=(U_{\overline{\theta}}f, (z-P_{\theta})^{-1}U_{\theta}g)$ for real $\theta$ and $\mathrm{Im}z>0$. 
The right hand side is analytic with respect to $\theta$ by Proposition 2.1. 
This in turn implies that the left hand side has a meromorphic continuation to 
$\mathrm{Im}z>-\beta \delta_1$ by Proposition 2.2. 
Then we have 
\[
\frac{1}{2\pi i}\oint_z(f, R_{+}(\zeta)g)d\zeta=\frac{1}{2\pi i}\oint_z 
(U_{\overline{\theta}}f, (\zeta-P_{\theta})^{-1}U_{\theta}g)d\zeta
=(U_{\overline{\theta}}f, \Pi_z^{\theta}U_{\theta}g).
\]
We note that if we replace $\phi(x)$ by $\varepsilon^n \phi(\varepsilon x)$ 
in the definition of $F(x)$, $v(x)$ and $P_{\theta}$, the Lipschitz constant of $v(x)$ is bounded by 
$C\varepsilon$ for some $C>0$. 
Thus taking $\varepsilon>0$ sufficiently small and arguing as in \cite[Theorem 3]{Hu}, 
we see that $\{U_{\theta}f\mid f\in \mathcal{A}\}$ is dense in $L^2$.
These prove the Proposition.
\end{proof}

\begin{corollary}
In addition to Assumption 1, suppose that $V$ has an analytic continuation to 
$|\mathrm{Im}z|<\delta_0$ and is bounded in this region. 
Then for $-\delta_0<\mathrm{Im}\theta<0$, the resonances of $P$ in $\mathrm{Im}z>\beta \mathrm{Im}\theta$
coincide with the eigenvalues of $P_{\theta}'=-\hbar^2\Delta+\beta x_1+\beta \theta +V(x_1+\theta, x' )$ 
including multiplicities.
In particular, $\mathrm{Res}(-\hbar^2 \Delta+\beta x_1)=\emptyset$.
\end{corollary}
\begin{proof}
Arguing as above, the eigenvalues of $P_{\theta,}'$ are described by the meromorphic continuation of 
$(f, R_{+}(z)g)$ for $f,g\in \mathcal{A}$ and thus coincide with $\mathrm{Res}(P)$ by Proposition 2.4.
\end{proof}

\section{Non-trapping estimates}
\begin{proof}[Proof of Theorem 2]
We only sketch the proof since it is similar to that of \cite[Theorem 1]{SZ}.
The non-trapping assumption enables us to construct an escape function $G\in C_c^{\infty}(T^*\mathbb{R}^n)$ 
such that $\{p, G\}\ge 1$ on $p^{-1}([\widetilde{a}, \widetilde{b}])\cap \{|x|<R\}$ 
for some $\widetilde{a}<a<b<\widetilde{b}$, where $R>0$ is large. 
We set $P_{\theta, \varepsilon}=e^{-\varepsilon G^W/{\hbar}}P_{\theta}e^{\varepsilon G^W/{\hbar}}$, 
where $M_1\hbar \le \varepsilon \ll |\mathrm{Im}\theta|$ and $M_1 \gg 1$. 
We consider $z$ with $a\le \mathrm{Re}z \le b$ and $(\mathrm{Im}z)_-\ll \varepsilon$.  

Take microlocal cutoffs $\Psi_1$, $\Psi_2$ and $\Psi_3$ near 
$\{x_1 \ge R_1\}\cup \{|x_1|<R_1, |x'|<R', p(x, \xi) \not \in [\widetilde{a}, \widetilde{b}]\}, 
\{|x_1|<R_1, |x'|<R',  p(x, \xi) \in [\widetilde{a}, \widetilde{b}]\}$ and 
$\{x_1<-R_1\}\cup \{|x_1|<R_1, |x'|>R'\}$ respectively, where $1\ll R_1 \ll R' \ll R$.  
The elliptic estimate implies 
$\|(P_{\theta, \varepsilon}-z)\Psi_1 u\|\ge c\|\Psi_1 u\|-\mathcal{O}(\hbar^{\infty})\|u\|$ 
for $R_1\gg 1$. 
Lemma 2.1, the construction of $G$ and the sharp G\r{a}rding inequality imply that 
$\|(P_{\theta, \varepsilon}-z)\Psi_2 u\|
\ge c\varepsilon\|\Psi_2 u\|-\mathcal{O}(\hbar^{\infty})\|u\|$ 
for $M_1\gg 1$ and $(\mathrm{Im}z)_-\ll \varepsilon$. 
Since $P_{\theta, \varepsilon}$ is not elliptic in the semiclassical sense, 
we estimate $\Psi_3 u$ by considering quadratic form. 
Then Lemma 2.1 implies  
$\|(P_{\theta, \varepsilon}-z)\Psi_3 u\| 
\ge c|\mathrm{Im}\theta| \|\Psi_3 u\|$
for $(\mathrm{Im}z)_-\ll \varepsilon \ll |\mathrm{Im}\theta|$. 
Thus 
\begin{align*}
\|u\|
&\le C\varepsilon^{-1}\sum \|(P_{\theta, \varepsilon}-z)\Psi_j u\| \\
&\le C\varepsilon^{-1} \|(P_{\theta, \varepsilon}-z)u\|+C\varepsilon^{-1} \sum \|[P_{\theta, \varepsilon}, \Psi_j] u\| \\
&\le C\varepsilon^{-1} \|(P_{\theta, \varepsilon}-z)u\|+C\hbar/{\varepsilon}(\|(P_{\theta, \varepsilon}-z)u\|+\|u\|).
\end{align*}
Choosing $M_1>0$ large and substituting $C\hbar/{\varepsilon}\|u\|<1/2\|u\|$, 
we obtain $\|(P_{\theta, \varepsilon}-z)u\|\ge c\varepsilon \| u\|$.

For $(\mathrm{Im}z)_- \le M_1\hbar $, we take $\varepsilon=\widetilde{M_1} \hbar$ with $\widetilde{M_1} \gg M_1$ and we have 
$ \|(P_{\theta}-z)^{-1}\|\le C\hbar^{-1}\le  C\exp(C(\mathrm{Im}z)_-/\hbar)/\hbar$ since 
$\|e^{\pm\varepsilon G^W/{\hbar}}\|\le C$. 

For $M_1\hbar \le (\mathrm{Im}z)_- \le  M\hbar\log \hbar^{-1} $,  we take
$\varepsilon=C(\mathrm{Im}z)_-$ with large $C>0$ and we have 
$\|(P_{\theta}-z)^{-1}\|\le C\exp(C\varepsilon/\hbar)/\varepsilon
\le  C\exp(C(\mathrm{Im}z)_-/\hbar)/(\mathrm{Im}z)_-
\le  C\exp(C(\mathrm{Im}z)_-/\hbar)/\hbar$ 
since $\|e^{\pm\varepsilon G^W/{\hbar}}\|\le \exp(C\varepsilon/\hbar)$. 
\end{proof}

\begin{corollary}
Suppose that Assumption 1 holds and $K_{[a,b]}=\emptyset$. 
Then for any $\psi \in C_c^{\infty}([a,b])$ and $\chi \in L_{\mathrm{cone}}^{\infty}(\mathbb{R}^n)$, 
there exists $C>0$ such that 
\[ \chi e^{-itP/{\hbar}}\psi(P) \chi=\mathcal{O}_{L^2 \to L^2}(\langle(t-C)_+/{\hbar}\rangle^{-\infty}), \]
where $(t-C)_+=\max\{t-C, 0\}$.
\end{corollary}
\begin{proof}
This follows from Theorem 2 employing Stone's formula, an almost analytic extension of $\psi$ 
and Green's formula. 
Since the proof is the same as that of \cite[Lemma 4.2]{NSZ}, we omit the details.
\end{proof}

\section{Shape resonance model}
 In this section, we discuss the shape resonances for the Stark Hamiltonian generated by potential wells. 
Recall that $p(x,\xi)=|\xi|^2+V_{\beta}(x)$, $V_{\beta}=\beta x_1+V$ 
and $K_{[a,b]}$ is the trapped set in the energy interval $[a,b]$. 
Throughout this section, we assume Assumption 1 and Assumption 2. 
Note that Assumption 2 implies $K_{[a,b]}=\{(x,\xi)|x\in \mathcal{G}^{\mathrm{int}}, a\le p(x,\xi)\le b\}$. 
We fix sufficiently small $\delta>0$. 
Then Assumption 2 holds true with $[a,b]$ replaced by $[a-\delta, b+\delta]$.

Fix a cutoff function $\chi_0$ near $\mathcal{G}^{\mathrm{int}}$ such that 
$\mathrm{supp}\partial \chi_0 \Subset \{x\in \mathbb{R}^n|V(x)>b+2\delta\}$. 
Complex distorted operators in this section are constructed outside $\mathrm{supp}\chi_0$. 
Let $V^{\mathrm{ext}}(x)$ be a potential obtained by filling up the wells: 
$V^{\mathrm{ext}}=V_{\beta}$ near $\mathrm{supp}(1-\chi_0)$ and 
$V^{\mathrm{ext}}>b+2\delta$ near $\mathcal{G}^{\mathrm{int}}$, 
and $P^{\mathrm{ext}}=-\hbar^2\Delta+V^{\mathrm{ext}}$ 
with the corresponding distorted operator $P^{\mathrm{ext}}_{\theta}$.
Let $V^{\mathrm{int}}(x)$ be a potential flattened outside the wells: 
$V^{\mathrm{int}}(x)=V_{\beta}$ near $\mathrm{supp}\chi_0$ 
and $V^{\mathrm{int}}(x)=b+2\delta$ outside a small neighborhood of $\mathrm{supp}\chi_0$,  
and $P^{\mathrm{int}}=-\hbar^2\Delta+V^{\mathrm{int}}$.

In the following we set $\alpha (\hbar)=\hbar^C$ and $\gamma(\hbar)=M\hbar\log\hbar^{-1}$, or 
$\alpha (\hbar)=C\hbar$ and $\gamma(\hbar)=M\hbar$. 
Then Theorem 2 implies that 
$\|(P^{\mathrm{ext}}_{\theta}(\hbar)-z)^{-1}\|=\mathcal{O}(\alpha(\hbar)^{-1})$ 
for $a-\delta \le \mathrm{Re}z \le b+\delta$, $\mathrm{Im}z \ge -\gamma(\hbar)$ 
and $\theta =-i\widetilde{M}\hbar\log\hbar^{-1}$. 
\begin{remark}
The results in subsection 4.1 and 4.2 remain true 
if we replace the non-trapping condition outside the wells by a resolvent assumption as follows: 
there exist $\alpha (\hbar)$, $\gamma(\hbar)$ and real numbers $a<b$ with 
$\alpha (\hbar), \gamma(\hbar)>e^{-S/{\hbar}}$ for any $S>0$ such that 
$\|(P^{\mathrm{ext}}_{\theta}(\hbar)-z)^{-1}\|=\mathcal{O}(\alpha(\hbar)^{-1})$ 
for $a-\delta \le \mathrm{Re}z \le b+\delta$ and $\mathrm{Im}z \ge -\gamma(\hbar)$. 
\end{remark}
The basic estimate in this section is the following Agmon estimate 
which is valid in more general settings (see \cite[section 7.1]{Z}).
\begin{lemma}
For any open set $U$ with $\overline{U}\subset \{x\in \mathbb{R}^n|V(x)>b+2\delta\}$, 
any $z\in [b-C_0, b+\delta]+i[-C_0,C_0]$ and small $\hbar>0$, there exists $S_0>0$ such that 
\[ \|u\|_{H^2_{\hbar}(U)} \le e^{-S_0/{\hbar}}\|u\|_{L^2(U_1)}+C\|(P-z)u\|_{L^2(U_1)}, \]
where $U_1$ is any open set with $\overline{U}\subset U_1$. 
\end{lemma}
This is also valid for $P_{\theta}$ if $U$ is away from the region of deformation 
in the definition of $P_{\theta}$. 
In the following we fix $S_0$ such that  
Lemma 4.1 holds true where $U$ is a small neighborhood of $\mathrm{supp}\partial \chi_0$, 
and moreover Lemma 4.1 with $P$ replaced by $P^{\mathrm{int}}$ holds true  
where $U$ is a small neighborhood of $\mathrm{supp}(1-\chi_0)$.

\subsection{Resolvent estimate}
In \cite{NSZ} the resolvent estimate is obtained by the abstract method based on the maximum principle technique. 
In the shape resonance model, we give more direct resolvent estimate based on the commutator calculation and 
the Agmon estimate.
\begin{prop}
For small $\hbar>0$,  
\[\|(1-\chi_0)(P_{\theta}-z)^{-1}\|\le C\alpha(\hbar)^{-1}, \|\chi_0(P_{\theta}-z)^{-1}\|
\le C\mathrm{dist}(z, \sigma(P^{\mathrm{int}}))^{-1}\]
if $a-\delta \le \mathrm{Re}z \le b+\delta$, $\mathrm{Im}z \ge -\gamma(\hbar)$ 
and $\mathrm{dist}(z, \sigma(P^{\mathrm{int}}))\ge e^{-S_0/{\hbar}}$.
\end{prop}
\begin{proof}
We have 
\begin{align*}
\|(1-\chi_0)(P_{\theta}-z)^{-1}\|&
=\|(P_{\theta}^{\mathrm{ext}}-z)^{-1}(P_{\theta}^{\mathrm{ext}}-z)(1-\chi_0)(P_{\theta}-z)^{-1}\|\\
&\le\alpha(\hbar)^{-1}\|(P_{\theta}-z)(1-\chi_0)(P_{\theta}-z)^{-1}\|\\
&\le\alpha(\hbar)^{-1}(1+\|[P_{\theta}, \chi_0](P_{\theta}-z)^{-1}\|)\\ 
&\le C\alpha(\hbar)^{-1}(1+e^{-S_0/{\hbar}}\|(P_{\theta}-z)^{-1}\|)\\
&\le C\alpha(\hbar)^{-1}(1+e^{-S_0/{\hbar}}\|\chi_0(P_{\theta}-z)^{-1}\|).
\end{align*}
The third inequality follows from the Agmon estimate. 
The last inequality follows if we subtract $C\alpha(\hbar)^{-1}e^{-S_0/{\hbar}}\|(1-\chi_0)(P_{\theta}-z)^{-1}\|
\le \frac{1}{2}\|(1-\chi_0)(P_{\theta}-z)^{-1}\|$ from both sides for small $\hbar>0$.
We also have 
\begin{align*}
\|\chi_0(P_{\theta}-z)^{-1}\|&= \|(P^{\mathrm{int}}-z)^{-1}(P^{\mathrm{int}}-z)\chi_0(P_{\theta}-z)^{-1}\| \\
&\le \mathrm{dist}(z, \sigma(P^{\mathrm{int}}))^{-1}\|(P_{\theta}-z)\chi_0(P_{\theta}-z)^{-1}\| \\
&\le \mathrm{dist}(z, \sigma(P^{\mathrm{int}}))^{-1}(1+\|[P_{\theta},\chi_0](P_{\theta}-z)^{-1}\|) \\
&\le C\mathrm{dist}(z, \sigma(P^{\mathrm{int}}))^{-1}(1+\hbar e^{-S_0/{\hbar}}\|(P_{\theta}-z)^{-1}\|)\\
&\le C\mathrm{dist}(z, \sigma(P^{\mathrm{int}}))^{-1}(1+\hbar e^{-S_0/{\hbar}}\|(1-\chi_0)(P_{\theta}-z)^{-1}\|).
\end{align*}
The third inequality follows from the Agmon estimate.
The last inequality follows if we subtract
$C\hbar\mathrm{dist}(z, \sigma(P^{\mathrm{int}}))^{-1}e^{-S_0/{\hbar}}\|\chi_0(P_{\theta}-z)^{-1}\|\le 
C\hbar\|\chi_0(P_{\theta}-z)^{-1}\|$ from both sides for small $\hbar>0$.
Substituting the left hand side of each inequality for the right hand side of the other inequality 
and subtracting the small remainder from both sides, we obtain the desired results.
\end{proof}
\begin{remark}
This proposition shows the dichotomy for resonances: 
\[\mathrm{Res}(P(\hbar))\cap ([a-\delta, b+\delta]-i[e^{-S_0/{\hbar}}, \gamma(\hbar)])
=\emptyset\mspace{10mu}\text{for small}\mspace{7mu}\hbar>0.\]
\end{remark}

As in \cite{S2} and \cite{NSZ}, we decompose resonances into clusters.
\begin{lemma}
For small $\hbar>0$, there exist $a_j(\hbar)<b_j(\hbar)<a_{j+1}(\hbar)$ such that 
\[\left(\mathrm{Res}(P(\hbar))\cup 
\sigma(P^{\mathrm{int}})\right)\cap ([a-\frac{\delta}{2}, b+\frac{\delta}{2}]-i[0, e^{-S_0/{\hbar}}])
\subset \cup_{j=1}^{J(\hbar)}\Omega_j(\hbar),\]
where $\Omega_j(\hbar)=[a_j(\hbar), b_j(\hbar)]-i[0, e^{-S_0/{\hbar}}]$, 
$b_j-a_j\le C\hbar^{-n}e^{-S_0/{\hbar}}$, $a_{j+1}-b_j\ge 2e^{-S_0/{\hbar}}$, 
$a_1\in(a-\frac{2}{3}\delta, a-\frac{1}{3}\delta)$, $b_{J(\hbar)}\in (b+\frac{1}{3}\delta,b+\frac{2}{3}\delta)$ and 
$\mathrm{Res}(P)\cap(([a_1-c\hbar^n, a_1]-i[0, e^{-S_0/{\hbar}}])
\cup ([b_{J(\hbar)}, b_{J(\hbar)}+c\hbar^n]-i[0, e^{-S_0/{\hbar}}]))
=\emptyset$.
Moreover, 
\[\|(1-\chi_0)(P_{\theta}-z)^{-1}\|\le C\alpha(\hbar)^{-1},\mspace{7mu} z\in \partial \widetilde{\Omega}_j(\hbar),\]
where $\widetilde{\Omega}_j(\hbar)=[a_j(\hbar)-e^{-S_0/{\hbar}}, 
b_j(\hbar)+e^{-S_0/{\hbar}}]+i[-2e^{-S_0/{\hbar}}, e^{-S_0/{\hbar}}]$.
\end{lemma}
\begin{proof}
The first statement follows easily from the fact that 
$\# (\sigma(P^{\mathrm{int}})\cap [a-\delta, b+\delta])=\mathcal{O}(\hbar^{-n})$ and 
Proposition 4.1 (or Remark 2.3). 
The second statement follows from Proposition 4.1.
\end{proof}

\subsection{The Weyl law}
We prove Theorem 3 in this subsection.
Set $ \Pi_j^{\theta}=\frac{1}{2\pi i}\int_{\partial \widetilde{\Omega}_j} (z-P_{\theta})^{-1}dz$ and 
$ \Pi_j^{\mathrm{int}}=\frac{1}{2\pi i}\int_{\partial \widetilde{\Omega}_j} (z-P^{\mathrm{int}})^{-1}dz$. 
Since $\mathrm{supp}\chi_0\cap \mathrm{supp}(P_{\theta}-P^{\mathrm{int}})=\emptyset$, we have
\begin{align*}
 \Pi_j^{\theta}-\Pi_j^{\mathrm{int}}
=\frac{1}{2\pi i}\int_{\partial \widetilde{\Omega}_j}
(z-P_{\theta})^{-1}(1-\chi_0)(P_{\theta}-P^{\mathrm{int}})(z-P^{\mathrm{int}})^{-1}dz. 
\end{align*}
\begin{prop}
For any $0<S<S_0$, 
\[ \Pi_j^{\theta}= \Pi_j^{\mathrm{int}}+\mathcal{O}(e^{-S/{\hbar}}).\]
\end{prop}

\begin{remark}
In the decaying potential case, we immediately have 
\[\|(P_{\theta}-P^{\mathrm{int}})(z-P^{\mathrm{int}})^{-1}\|\le  e^{-S_0/{\hbar}}\|(z-P^{\mathrm{int}})^{-1}\|+C\le C\]
for $z \in \partial \widetilde{\Omega}_j$ 
by the Agmon estimate for $P^{\mathrm{int}}$ and $\mathrm{dist}(z, \sigma(P^{\mathrm{int}}))\ge e^{-S_0/{\hbar}}$ since 
$P_{\theta}-P^{\mathrm{int}}$ has bounded coefficients. This and Lemma 4.2 imply 
\[\| \Pi_j^{\theta} - \Pi_j^{\mathrm{int}}\|\le C|\partial \widetilde{\Omega}_j| \alpha(\hbar)^{-1} 
=\mathcal{O}(e^{-S/{\hbar}}).\] 
Since $P_{\theta}-P^{\mathrm{int}}$ has an unbounded coefficient in our case, we need additional arguments.
\end{remark}

\begin{proof}[Proof of Proposition 4.2]
Since $z-P^{\mathrm{int}}$ is elliptic near $\mathrm{supp}(P_{\theta}-P^{\mathrm{int}})$, 
\begin{align*}
\|&(P_{\theta}-P^{\mathrm{int}})(z-P^{\mathrm{int}})^{-1}\chi_0\|\\
&\le C\|(z-P^{\mathrm{int}})(P_{\theta}-P^{\mathrm{int}})(z-P^{\mathrm{int}})^{-1}\chi_0\|_{L^2\to H_{\hbar}^{-2}} \\
&= C \|[P^{\mathrm{int}}, P_{\theta}-P^{\mathrm{int}}](z-P^{\mathrm{int}})^{-1}\chi_0\|_{L^2\to H_{\hbar}^{-2}} \\
&\le C e^{-S_0/{\hbar}}\|(z-P^{\mathrm{int}})^{-1}\|\le C, 
\end{align*}
where the last two inequalities follow from the Agmon estimate for $P^{\mathrm{int}}$ 
and $\mathrm{dist}(z, \sigma(P^{\mathrm{int}}))\ge e^{-S_0/{\hbar}}$ 
(note that $[P^{\mathrm{int}}, P_{\theta}-P^{\mathrm{int}}]$ has bounded coefficients). 
This and Lemma 4.2 imply 
\[\| \left(\Pi_j^{\theta} - \Pi_j^{\mathrm{int}}\right) \chi_0\|\le C|\partial \widetilde{\Omega}_j| \alpha(\hbar)^{-1} 
=\mathcal{O}(e^{-S/{\hbar}}).\] 
Finally, we have $\|\Pi_j^{\theta}(1-\chi_0)\|\le C|\partial \widetilde{\Omega}_j| \alpha(\hbar)^{-1} 
=\mathcal{O}(e^{-S/{\hbar}})$ by Lemma 4.2, 
and $\|(1-\chi_0)\Pi_j^{\mathrm{int}}\|\le C\hbar^{-n}e^{-S_0/{\hbar}}
=\mathcal{O}(e^{-S/{\hbar}})$ by the Agmon estimate.
\end{proof}

\begin{proof}[Proof of Theorem 3]
Proposition 4.2 implies that $\mathrm{rank} \Pi_j^{\theta}=\mathrm{rank}\Pi_j^{\mathrm{int}}$ for small $\hbar>0$.
Thus the Weyl law for discrete eigenvalues of $P^{\mathrm{int}}$ completes the proof.
\end{proof}

\subsection{Resonance expansion}
We prove Theorem 4 in this subsection. Theorem 5 and Theorem 6 are used in this subsection. 
In the following, we take $\psi \in C_c^{\infty}([a,b])$ and 
$\chi \in C_b^{\infty}\cap L^{\infty}_{\mathrm{cone}}$ as in Theorem 4.
We take $a(\hbar)=a_1(\hbar)-\frac{c}{2}\hbar^n$, $b(\hbar)=b_{J(\hbar)}+\frac{c}{2}\hbar^n$ (see Lemma 4.2) 
and set $\Omega (\hbar)=[a(\hbar),b(\hbar)]-i[0, \hbar]$. 
We first prove Theorem 4 after large time $t>\hbar^{-n+1-\varepsilon}$ (see Burq-Zworski~\cite{BZ}).
\begin{prop}
Under the above notation and for any $\varepsilon>0$, 
\[\chi e^{-itP/{\hbar}}\chi \psi(P)=\sum_{z\in \mathrm{Res}(P(\hbar))\cap\Omega (\hbar)}
\mathrm{Res}_{\zeta=z}e^{-it\zeta/{\hbar}}\chi R_+(\zeta, \hbar)\chi\psi(P)+\mathcal{O}(\hbar^{\infty})\]
for $t>\hbar^{-n+1-\varepsilon}$.
\end{prop} 

\begin{proof}
This is proved by Stone's formula, the almost analytic extension technique and Green's formula. 
If we employ Proposition 4.1 as the resolvent estimate, the claimed result follows.
Since the argument of the proof is the same as \cite{BZ}, we omit the details.
We note that calculations involving the energy cutoff $\psi(P)$ are justified by Theorem 5.
\end{proof}

\begin{remark}
If we employ Remark 2.3 as the resolvent estimate, the result of Burq-Zworski~\cite{BZ} is 
obtained for the Stark Hamiltonian case. 
Namely, Proposition 4.3 remains true under Assumption 1 for $t>\hbar^{-L}$ for some choices of $\Omega(\hbar)$ and $L>0$. 
\end{remark}

We move to the proof of Theorem 4 up to large time $C\le t\le e^{S/{2\hbar}}$.
We first prepare the Agmon estimate for continuous spectrum (\cite[Lemma 4.3]{NSZ}):
\begin{lemma}
If $\widetilde{\chi}_0 \in C_c^{\infty}(\mathbb{R}^n)$ is a cutoff near $\mathrm{supp}\partial \chi_0$ and 
$\psi_1 \in C_c^{\infty}(\mathbb{R})$ is supported near $[a, b]$, 
\[\widetilde{\chi}_0 \psi_1 (P(\hbar)), \widetilde{\chi}_0\psi_1 (P^{\mathrm{int}}(\hbar)), 
\widetilde{\chi}_0\psi_1 (P^{\mathrm{ext}}(\hbar))=\mathcal{O}_{L^2 \to H_{\hbar}^2}(e^{-S_0/{2\hbar}}).\]
\end{lemma}

\begin{proof}
This follows from the Agmon estimate, the almost analytic extension technique and Green's formula.  
Since the proof is the same as \cite[Lemma 4.3]{NSZ}, we omit the details.
\end{proof}

We next compare the different quantum dynamics \cite[Lemma 4.4]{NSZ}. 
\begin{lemma}
For $\psi_1 \in C_c^{\infty}(\mathbb{R})$ supported near $[a, b]$ and $t\in \mathbb{R}$,
\[(1-\chi_0) e^{-itP/{\hbar}}\psi_1(P)\chi_0=\mathcal{O}(|t|e^{-S_0/{2\hbar}})+\mathcal{O}(\hbar^{\infty}),\]
\[\chi_0 e^{-itP/{\hbar}}\psi_1(P)=\chi_0 e^{-itP^{\mathrm{int}}/{\hbar}}\psi_1(P^{\mathrm{int}})+
\mathcal{O}(|t|e^{-S_0/{2\hbar}})+\mathcal{O}(\hbar^{\infty}),\]
\[(1-\chi_0) e^{-itP/{\hbar}}\psi_1(P)=(1-\chi_0) e^{-itP^{\mathrm{ext}}/{\hbar}}\psi_1(P^{\mathrm{ext}})
+\mathcal{O}(|t|e^{-S_0/{2\hbar}})+\mathcal{O}(\hbar^{\infty}).\]
\end{lemma}
\begin{proof}
The proof relies on Duhamel's formula as in \cite{NSZ}.
Lemma 4.3 implies that 
$(1-\chi_0)e^{-itP/{\hbar}}\psi_1(P)\chi_0$, 
$\chi_0 (e^{-itP/{\hbar}}\psi_1(P)-e^{-itP^{\mathrm{int}}/{\hbar}}\psi_1(P^{\mathrm{int}}))$ and 
$(1-\chi_0) (e^{-itP/{\hbar}}\psi_1(P)-e^{-itP^{\mathrm{ext}}/{\hbar}}\psi_1(P^{\mathrm{ext}}))$
applied by $i\hbar \partial_t-P$ from the left are $\mathcal{O}_{L^2 \to L^2}(e^{-S_0/{2\hbar}})$.

As for the initial values, we have 
$(1-\chi_0)\psi_1(P)\chi_0=\mathcal{O}(\hbar^{\infty})$ by Theorem 5, 
$\chi_0 (\psi_1(P)-\psi_1(P^{\mathrm{int}}))=\mathcal{O}(\hbar^{\infty})$ 
by Theorem 5 and the usual functional calculus for elliptic pseudodifferential operators, 
and $(1-\chi_0) (\psi_1(P)-\psi_1(P^{\mathrm{ext}}))=\mathcal{O}(\hbar^{\infty})$ 
by Theorem 6 (Theorem 6 is used only at this point).
\end{proof}
\begin{prop}
Under the above notation and for any $0<S<S_0$, 
\begin{align*}
\chi e^{-itP/{\hbar}}\chi \psi(P)=&\sum_{z\in \mathrm{Res}(P(\hbar))\cap\Omega (\hbar)}
\mathrm{Res}_{\zeta=z}e^{-it\zeta/{\hbar}}\chi_1 R_+(\zeta, \hbar)\chi_1 \psi(P)\\
&+\chi_2 \mathcal{O}(\langle(t-C)_+/{\hbar}\rangle^{-\infty})\chi_2 \psi(P)+\mathcal{O}(\hbar^{\infty})
\end{align*}
for $0 \le t \le e^{S/{2\hbar}}$, where $\chi_1=\chi \chi_0$ and $\chi_2=\chi(1-\chi_0)$.
\end{prop} 
\begin{proof}
We only sketch the proof since it is the same as \cite{NSZ}. 
Lemma 4.4 and Theorem 5 show that $\chi e^{-itP/{\hbar}}\chi \psi(P)= 
\chi_1 e^{-itP^{\mathrm{int}}/{\hbar}}\psi_1(P^{\mathrm{int}})\chi_1 \psi(P)+
\chi_2 e^{-itP^{\mathrm{ext}}/{\hbar}}\psi_1(P^{\mathrm{ext}})\chi_2 \psi(P)+\mathcal{O}(\hbar^{\infty})$, where 
$\psi_1 \psi =\psi$. The second term is estimated by Corollary 3.1. 
The eigenfunction expansion of the first term is approximated by the first term of the 
right hand side of Proposition 4.4 by the same argument as in Proposition 4.2 
with $ \Pi_j^{\theta}=\frac{1}{2\pi i}\int_{\partial \widetilde{\Omega}_j} (z-P_{\theta})^{-1}dz$ and 
$ \Pi_j^{\mathrm{int}}=\frac{1}{2\pi i}\int_{\partial \widetilde{\Omega}_j} (z-P^{\mathrm{int}})^{-1}dz$ 
replaced by 
$ \frac{1}{2\pi i}\int_{\partial \widetilde{\Omega}_j} e^{-itz/{\hbar}}(z-P_{\theta})^{-1}dz$ and 
$ \frac{1}{2\pi i}\int_{\partial \widetilde{\Omega}_j} e^{-itz/{\hbar}}(z-P^{\mathrm{int}})^{-1}dz$ 
respectively.
\end{proof}
We next estimate the residue outside the well;
\begin{lemma}
For any $\widetilde{\chi} \in C_b^{\infty}\cap L^{\infty}_{\mathrm{cone}}$ and any $0<S<S_0$, 
\[\sum_{z\in \mathrm{Res}(P(\hbar))\cap\Omega_j (\hbar)}
\mathrm{Res}_{\zeta=z}e^{-it\zeta/{\hbar}}\chi_2 R_+(\zeta)\widetilde{\chi}=\mathcal{O}(e^{-S/{\hbar}})\]
for $0\le t\le e^{S_0/{\hbar}}$, where $\chi_2$ is as in Proposition 4.4.
\end{lemma}
\begin{proof}
 Since $|e^{-itz/{\hbar}}|$ is bounded on $\partial \widetilde{\Omega}_j$ for $0\le t\le e^{S_0/{\hbar}}$, 
we have by Lemma 4.2
\[\|\frac{1}{2\pi i}\int_{\partial \widetilde{\Omega}_j}
e^{-itz/{\hbar}}\chi(1-\chi_0) (z-P_{\theta})\widetilde{\chi}dz\|
\le C\alpha(\hbar)^{-1}
|\partial\widetilde{\Omega}_j|=\mathcal{O} (e^{-S/{\hbar}}).\] 
\end{proof} 
\begin{proof}[Proof of Theorem 4]
Proposition 4.3 proves Theorem 4 for $t>\hbar^{-n+1-\varepsilon}$. 
Proposition 4.4 and Lemma 4.5 prove Theorem 4 for $C\le t \le e^{S/{2\hbar}}$.
\end{proof}

\section{Functional pseudodifferential calculus in the Stark effect}
In this section, we prove Theorem 5 and Theorem 6. 
In subsection 5.1 and subsection 5.2, 
we set $P(\hbar)=-\hbar^2 \Delta + \beta  x_1+V(x)$, 
where $V \in C_b^{\infty}(\mathbb{R}^n ; \mathbb{R})$. 
The commutator calculations below are justified by Corollary A.1 in the Appendix. 
\subsection{Weighted resolvent estimates}
We estimate the weighted resolvents in this subsection. 
Take $w \in  C^{\infty}(\mathbb{R}^n ; \mathbb{R}_{\ge 1})$ depending only on $x_1$ and 
$w=|x_1|$ for $x_1 \le -2$ and $w=1$ for $x_1 \ge -1$.
\begin{lemma}
For any $k \ge 0$, $|z|\lesssim 1$ and $0<\hbar\le 1$, 
\[\|w^{-k-1} (P-z)^{-1}w^{k}\|_{L^2 \to H_{\hbar}^2} 
\lesssim |\mathrm{Im}z|^{-1} \left( 1+{\hbar}/|{\mathrm{Im}z}|\right)^{3k}.\]
\end{lemma}
\begin{proof}
We first prove the case where $k=0$.
Take $\chi \in  C^{\infty}(\mathbb{R}^n)$ depending only on $x_1$ and $\chi=0$ for $x_1 \le 1$ and $\chi=1$ for $x_1 \ge 2$. 
We set $\chi_R(x)=\chi(x/R)$.
\begin{align*}
&\|\langle hD \rangle^2 w^{-1}(P-z)^{-1}u\|_{L^2}\\
&\le C\|w^{-1}\langle hD \rangle^2 (P-z)^{-1}u\|_{L^2} \\
&=C\|w^{-1}(P-z+z-\beta x_1-V+1) (P-z)^{-1}u\|_{L^2}\\
&\le C\|\chi_R x_1 (P-z)^{-1}u\|_{L^2}+C\|u\|_{L^2}+C_R\|(P-z)^{-1}u\|_{L^2}\\
&\le C\|\chi_R x_1 (P-z)^{-1}u\|_{L^2}+C_R|\mathrm{Im}z|^{-1}\|u\|_{L^2},
\end{align*}
since $|z|\lesssim 1$. Since $P(\hbar)-z$ is elliptic near the support of $\chi$, we have 
\begin{align*} 
\|\chi_R x_1 (P-z)^{-1}u\|_{L^2}&\le C \|(P-z)\chi_R (P-z)^{-1}u\|_{L^2}\\
&\le C \|u\|_{L^2}+C\|[P, \chi_R](P-z)^{-1}u\|_{L^2}.
\end{align*}
Substituting $\|[P, \chi_R](P-z)^{-1}u\|_{L^2}\le C\hbar R^{-1}\|\langle hD \rangle^2 w^{-1}(P-z)^{-1}u\|_{L^2}$ 
for large $R$, the proof for $k=0$ is completed. 

We next assume that Lemma 5.1 is true for $k-1$. The case where $k=0$ implies 
\begin{align*}
&\|w^{-k-1} (P-z)^{-1}w^{k}\|_{L^2 \to H_{\hbar}^2}\\
&=\|w^{-1}(P-z)^{-1}(P-z)w^{-k} (P-z)^{-1}w^{k}\|_{L^2 \to H_{\hbar}^2}\\
&\lesssim|\mathrm{Im}z|^{-1}\|(P-z)w^{-k} (P-z)^{-1}w^{k}\|_{L^2 \to L^2} \\
&\lesssim|\mathrm{Im}z|^{-1}+ (\hbar/|\mathrm{Im}z|)\|w^{-k-1} (P-z)^{-1}w^{k}\|_{L^2 \to H_{\hbar}^1}.
\end{align*}
We have
\begin{align*}
&(\hbar/|\mathrm{Im}z|)\|w^{-k-1} (P-z)^{-1}w^{k}\|_{L^2 \to H_{\hbar}^1}
\le (\hbar/|\mathrm{Im}z|)\|w^{-1} (P-z)^{-1}\|_{L^2 \to H_{\hbar}^1}\\
&+(\hbar/|\mathrm{Im}z|)\|w^{-k-1} (P-z)^{-1}[P, w^k](P-z)^{-1}\|_{L^2 \to H_{\hbar}^1}.
\end{align*}
The first term can be estimated by $|\mathrm{Im}z|^{-1}(\hbar/|\mathrm{Im}z|)$ by the case where $k=0$. 
The second term can be estimated by 
\begin{align*}
&(\hbar/{|\mathrm{Im}z|})^2\|w^{-k-1} (P-z)^{-1}\langle \hbar D \rangle w^{k-1}\|_{L^2 \to H_{\hbar}^1}\\
&\lesssim (\hbar/{|\mathrm{Im}z|})^2\|w^{-2} (P-z)^{-1} \|_{L^2 \to H_{\hbar}^2}\\
&+(\hbar/{|\mathrm{Im}z|})^2\|w^{-k-1} (P-z)^{-1}[P, \langle \hbar D \rangle w^{k-1}](P-z)^{-1}\|_{L^2 \to H_{\hbar}^1}.
\end{align*}
The first term can be estimated by $|\mathrm{Im}z|^{-1}(\hbar/{|\mathrm{Im}z|})^2$ by the case where $k=0$. 
The second term can be estimated by 
\begin{align*}
&(\hbar/{|\mathrm{Im}z|})^3\|w^{-k-1} (P-z)^{-1}w^{k-1}\|_{L^2 \to H_{\hbar}^1}\\
&+(\hbar/{|\mathrm{Im}z|})^2\|w^{-k-1} (P-z)^{-1}w^{k-1}\|_{L^2 \to H_{\hbar}^1}
\cdot\|w^{-1}\langle \hbar D \rangle^2\hbar(P-z)^{-1}\|_{L^2 \to L^2}\\
&\lesssim (\hbar/{|\mathrm{Im}z|})^3\|w^{-k-1} (P-z)^{-1}w^{k-1}\|_{L^2 \to H_{\hbar}^1}
\end{align*}
by the case where $k=0$. The induction hypothesis completes the proof.
\end{proof}
\begin{remark}
Similar calculations show that 
\[\|w^k (P-z)^{-1}w^{-k}\|_{L^2 \to L^2} \lesssim |\mathrm{Im}z|^{-1} \left(1+ {\hbar}/|{\mathrm{Im}z}| \right)^{2k} \]
and 
\[\|w^{k-1} (P-z)^{-1}w^{-k}\|_{L^2 \to H_{\hbar}^2} 
\lesssim |\mathrm{Im}z|^{-1} \left( 1+{\hbar}/|{\mathrm{Im}z}|\right)^{2k}\]
for $|z|\lesssim 1$ and $0<\hbar \le 1$.
\end{remark}

\subsection{Weighted resolvents as $\Psi$DOs}
We set 
\[S_{\delta}(m)=\{a(\bullet ;\hbar) \in C^{\infty}(T^*\mathbb{R}^n) |
 |\partial_{x, \xi}^{\alpha} a(x, \xi; \hbar)|
\le C_{\alpha} \hbar^{-\delta |\alpha|} m(x, \xi) \}.\]
The natural asymptotic expansion for $a\in S_{\delta}(m)$ with $0\le \delta <\frac{1}{2}$ 
is of the form $a \sim \sum \hbar^{(1-2\delta)j}a_j$ with $a_j\in S_{\delta}(m)$. 
We set $S_{\delta}(m_1m_2^{-\infty})=\bigcap_{N>0}S_{\delta}(m_1m_2^{-N})$. 

To simplify the statement, 
we introduce the symbol class for weighted resolvents 
$S^{-k}_{\mathrm{WR}}(m)=|\mathrm{Im}z|^{-k}S^0_{\mathrm{WR}}(m)$, 
where
\begin{align*}
&S^0_{\mathrm{WR}}(m)=\{a(x, \xi; z, \hbar)| 
|\partial_{x, \xi}^{\alpha}a|\le C_{\alpha}|\mathrm{Im}z|^{-C_{\alpha}}m(x, \xi)
\mspace{7mu} \text{for} \mspace{7mu} |z|\lesssim 1 \mspace{7mu} \text{and} \\
&\mspace{7mu} a\in S_{\delta}(m) 
\mspace{7mu} \text{uniformly for} \mspace{7mu} \hbar^{\delta}\lesssim|\mathrm{Im}z|, |z|\lesssim 1 
\mspace{7mu} \text{for any} \mspace{7mu}0\le\delta<\frac{1}{2}\}.
\end{align*}
We say that $a\in S^{-k}_{\mathrm{WR}}(m)$ has an asymptotic expansion 
$a\sim \sum \hbar^j a_j$ in $S^{-k}_{\mathrm{WR}}(m)$ if $a_j \in S^{-k-2j}_{\mathrm{WR}}(m)$ and 
$a\sim \sum \hbar^j a_j=\hbar^{-k\delta}\sum \hbar^{(1-2\delta)j} \hbar^{(k+2j)\delta}a_j$ 
in $\hbar^{-k\delta}S_{\delta}(m)$ uniformly for 
$\hbar^{\delta}\lesssim|\mathrm{Im}z|, |z|\lesssim 1$ for any $0\le\delta<\frac{1}{2}$. 
We set $S^{-k}_{\mathrm{WR}}(m_1m_2^{-\infty})=\bigcap_{N>0}S^{-k}_{\mathrm{WR}}(m_1m_2^{-N})$. 

In the following, we set $m=|\xi|^2+\langle x_1 \rangle$.
\begin{prop}
If $b\in S^0_{\mathrm{WR}}(w^{-\infty}m^{-k}\langle x' \rangle ^{-s'})$, 
then 
\[(P-z)^{-1}b^W \in \mathrm{Op}S^{-1}_{\mathrm{WR}}(w^{-\infty}m^{-k-1}\langle x' \rangle ^{-s'}).\] 
\end{prop}
\begin{proof}
We set $\widetilde{P}=-\hbar^2 \Delta+\beta\langle x_1 \rangle+C$, 
where $C \gg 1$ so that $\widetilde{P}^{-1}\in \mathrm{Op}S(m^{-1})$. 
Applying $\langle x' \rangle ^{s'}\widetilde{P}^{k}$ from the right, we may assume that $s'=k=0$. 
Applying $w^j \widetilde{P}$ from the right, we only have to prove 
$(P-z)^{-1}b^W\widetilde{P} \in \mathrm{Op}S^{-1}_{\mathrm{WR}}(1)$. 
Since $\widetilde{P}\sim P+2\beta w$, 
we only have to prove  $(P-z)^{-1}b^W(P-z)=b^W+(P-z)^{-1}[b^W, P] \in \mathrm{Op}S^{-1}_{\mathrm{WR}}(1)$ 
and $(P-z)^{-1}b^W \in \mathrm{Op}S^{-1}_{\mathrm{WR}}(1)$. 
For this it is enough to prove $(P-z)^{-1}\langle \hbar D \rangle b^W \in \mathrm{Op}S^{-1}_{\mathrm{WR}}(1)$.
Let $l_1, l_2, \dots ,l_N$ be linear forms on $\mathbb{R}^{2n}$. 
Then $ad_{l_1^W(x, \hbar D)}\dots ad_{l_N^W(x, \hbar D)}\left((P-z)^{-1}\langle \hbar D \rangle b^W \right)$ 
consists of the terms such as 
\begin{align*}
&(P-z)^{-1}(ad_{l_1^W(x, \hbar D)}P)(P-z)^{-1}(ad_{l_2^W(x, \hbar D)}P)(P-z)^{-1}
 (ad_{l_3^W(x, \hbar D)}\\&ad_{l_4^W(x, \hbar D)}P)(P-z)^{-1}  
 \dots (ad_{l_{N-1}^W(x, \hbar D)}P)(P-z)^{-1}ad_{l_N^W(x, \hbar D)}(\langle \hbar D \rangle b^W)\\
&=((P-z)^{-1}(ad_{l_1^W(x, \hbar D)}P)w^{-1})(w(P-z)^{-1} (ad_{l_2^W(x, \hbar D)}P)w^{-2})
 (w^2\\&(P-z)^{-1}(ad_{l_3^W(x, \hbar D)}ad_{l_4^W(x, \hbar D)}P)w^{-3})\dots 
 (w^{s-1}(P-z)^{-1}(ad_{l_{N-1}^W(x, \hbar D)}P)\\
&w^{-s})(w^s(P-z)^{-1}\langle \hbar D \rangle w^{-s-1})(w^{s+1}\langle \hbar D \rangle^{-1}
 ad_{l_N^W(x, \hbar D)}(\langle \hbar D \rangle b^W)),
\end{align*}
where $s\le N$. Lemma 5.1 and Beals's theorem complete the proof.
\end{proof}

We next calculate the asymptotic expansion of the weighted resolvent. 
Let $r(x, \xi, z, \hbar)\sim \sum_{j \ge 0}\hbar^j r_j$ be the formal symbol  
of $(P-z)^{-1}$ given by the standard parametrix construction, which does not belong to any symbol class. 
We easily see that $r_0=(p(x, \xi)-z)^{-1}$ and $r_j(x, \xi, z)=\frac{q_j(x, \xi, z)}{(p(x, \xi)-z)^{2j+1}}$ for $j \ge 1$, 
where $q_j(x, \xi, z)=\sum_{k=0}^{2j-1}q_{j,k}(x,\xi)z^k$ with $q_{j,k}(x,\xi)\in S(m^{2j-k})$. 
\begin{prop}
Suppose that $b$ has an asymptotic expansion $\sim \sum \hbar^j b_j$ in 
$S^0_{\mathrm{WR}}(w^{-\infty}m^{-k}\langle x' \rangle^{-s'})$. 
Then the symbol of $(P-z)^{-1}b^W$ has an asymptotic expansion 
$\sim (\sum \hbar^jr_j) \sharp (\sum \hbar^j b_j)$ 
in $S^{-1}_{\mathrm{WR}}(w^{-\infty}m^{-k-1}\langle x' \rangle^{-s'})$. 
\end{prop} 
\begin{proof}
Take $0\le\delta<\frac{1}{2}$ and consider $z$ with 
$\hbar^{\delta}\lesssim|\mathrm{Im}z|, |z|\lesssim 1$. Borel's theorem enables us to take 
$a\in \hbar^{-\delta}S_{\delta}(w^{-\infty}m^{-k-1}\langle x' \rangle^{-s'})$ 
such that $a$ has an asymptotic expansion 
$a\sim \hbar^{-\delta}(\sum_j \hbar^{(1-2\delta)j} \hbar^{(2j+1)\delta }r_j) 
\sharp (\sum \hbar^{(1-2\delta)j} \hbar^{2j\delta}b_j)$ 
in $\hbar^{-\delta}S_{\delta}(w^{-\infty}m^{-k-1}\langle x' \rangle^{-s'})$ which is uniform with respect to $z$. 
Then $(P-z)a^W=b^W+\hbar^{\infty}\mathrm{Op}S(w^{-\infty}m^{-k}\langle x' \rangle^{-s'})$ 
since $(p-z)\sharp ((\sum \hbar^j r_j) \sharp (\sum \hbar^j b_j)) 
\sim ((p-z)\sharp  (\sum \hbar^j r_j)) \sharp (\sum \hbar^j b_j) \sim \sum \hbar^j b_j$ in the formal power series sense.
Thus, 
\begin{align*}
a^W(x, \hbar D; \hbar)&=(P-z)^{-1}b^W +(P-z)^{-1}\hbar^{\infty}\mathrm{Op}S(w^{-\infty}m^{-k}\langle x' \rangle^{-s'})\\
&=(P-z)^{-1}b^W +\hbar^{\infty}\mathrm{Op}S(w^{-\infty}m^{-k-1}\langle x' \rangle^{-s'}).
\end{align*}
The last equality follows from Proposition 5.1.
\end{proof}

\subsection{Proofs}
\begin{proof}[Proof of Theorem 5]
Applying $\langle x' \rangle^{s'}$ from the right, we may assume that $s'=0$. 
We take an almost analytic extension $\widetilde{f} \in C_c^{\infty}(\mathbb{C})$ of $f$: 
$\overline{\partial} \widetilde{f}=\mathcal{O}(|\mathrm{Im}z|^{\infty})$ and $\widetilde{f}|_{\mathbb{R}}=f$. 
The Helffer-Sj\"{o}strand formula shows 
\[f(P)\chi^W=\frac{1}{2\pi i}\int \overline{\partial}\widetilde{f}(z)(z-P)^{-1}\chi^Wdz\wedge d\overline{z}.\]
Take $0<\delta<\frac{1}{2}$. 
Proposition 5.1 implies $(z-P)^{-1}\chi^W\in \mathrm{Op}S^{-1}_{\mathrm{WR}}(w^{-\infty}m^{-1})$. 
Thus $f(P)\chi^W=a^W(x, \hbar D; \hbar) \in \mathrm{Op}S(w^{-\infty}m^{-1})$ and 
the integral for $|\mathrm{Im}z|<h^{\delta}$ contributes only as $h^{\infty}\mathrm{Op}S(w^{-\infty}m^{-1})$. 
Proposition 5.2 implies that $(z-P)^{-1}\chi^W$ has an asymptotic expansion in 
$\hbar^{-\delta}S_{\delta}(w^{-\infty}m^{-1})$ which is uniform with respect to $z$ with $|\mathrm{Im}z|>h^{\delta}$. 
Thus 
$a\sim \left( \hbar^{-\delta}\sum \hbar^{j(1-2\delta )} \hbar^{(1+2j)\delta}\widetilde{a}_j\right)\sharp \chi$ 
in $\hbar^{-\delta}S_{\delta}(w^{-\infty}m^{-1})$, where 
\[\widetilde{a}_j=\frac{1}{2\pi i}\int_{|\mathrm{Im}z|>h^{\delta}} 
     \overline{\partial}\widetilde{f}(z)\frac{q_j(x, \xi, z)}{(z-p(x, \xi))^{2j+1}}dz\wedge d\overline{z}.\]
We set 
\[a_j=\frac{1}{2\pi i}\int \overline{\partial}\widetilde{f}(z)
\frac{q_j(x, \xi, z)}{(z-p(x, \xi))^{2j+1}}dz\wedge d\overline{z}
      =\frac{1}{(2j)!}\partial_t^{2j}(q_j(x, \xi, t)f(t))_{t=p(x, \xi)}.\]
We easily see that $(a_j-\widetilde{a}_j)\sharp\chi \in \hbar^{\infty}S(w^{-\infty}m^{-1})$ and 
$a_j \in  S(w^{-\infty}m^{-\infty})$. 
Thus we have in fact 
$a\sim \left( \sum \hbar^j a_j\right)\sharp \chi$ in $S(w^{-\infty}m^{-1})$. 
We set $f_k(t)=(t-i)^{k} f(t)$. 
Then $f_k(P)\chi^W$ has an asymptotic expansion in $S(w^{-\infty}m^{-1})$ by the above argument. 
Proposition 5.2 with $z=i$ implies that 
$f(P)\chi^W=(P-i)^{-k}f_k(P)\chi^W$ has an asymptotic expansion in $S(w^{-\infty}m^{-k-1})$, 
which coincides with the formal one $\left( \sum \hbar^j a_j\right)\sharp \chi$. 
Since $k$ is arbitrary, $f(P)\chi^W$ has an asymptotic expansion in $S(w^{-\infty}m^{-\infty})=S(m^{-\infty})$. 
\end{proof}

\begin{proof}[Proof of Theorem 6]
The Helffer-Sj\"{o}strand formula and the resolvent equation show that 
\[f(P_2)-f(P_1)
=\frac{1}{2\pi i}\int \overline{\partial}\widetilde{f}(z)(z-P_2)^{-1}(V_2-V_1)(z-P_1)^{-1}dz\wedge d\overline{z}.\]
Take $0<\delta<\frac{1}{2}$. 
We have $(V_2-V_1)(z-P_1)^{-1}\in \mathrm{Op}S^{-1}_{\mathrm{WR}}(w^{-\infty}m^{-1}\langle x' \rangle ^{-s'})$ 
by Proposition 5.1. 
Thus Proposition 5.1 again implies that $(z-P_2)^{-1}(V_2-V_1)(z-P_1)^{-1} 
\in \mathrm{Op}S^{-2}_{\mathrm{WR}}(w^{-\infty}m^{-2}\langle x' \rangle ^{-s'})$. 
This implies that $f(P_2)-f(P_1)\in \mathrm{Op}S(w^{-\infty}m^{-2}\langle x' \rangle ^{-s'})$ and 
the integral for $|\mathrm{Im}z|<h^{\delta}$ contributes only as 
$h^{\infty}\mathrm{Op}S(w^{-\infty}m^{-2}\langle x' \rangle ^{-s'})$. 
The twice applications of Proposition 5.2 show that 
$(z-P_2)^{-1}(V_2-V_1)(z-P_1)^{-1}$ has an asymptotic expansion 
which is uniform with respect to $z$ with $|\mathrm{Im}z|>h^{\delta}$ 
in $\hbar^{-2\delta}S_{\delta}(w^{-\infty}m^{-2}\langle x' \rangle ^{-s'})$. 
Thus the similar calculation as in the proof of Theorem 5 based on 
the partial fraction expansion shows that $f(P_2)-f(P_1)$ has an asymptotic expansion in 
$\mathrm{Op}S(w^{-\infty}m^{-2}\langle x' \rangle^{-s'})$. 
We next prove that $f(P_2)-f(P_1)$ has an asymptotic expansion 
in $\mathrm{Op}S(w^{-\infty}m^{-N}\langle x' \rangle^{-s'})$ for any $N$. 
Suppose that this is true for $N$. 
Applying this to $g(t)=(t+i)f(t)$, 
we see that $(P_2+i)f(P_2)-(P_1+i)f(P_1)$ has an asymptotic expansion in 
$\mathrm{Op}S(w^{-\infty}m^{-N}\langle x' \rangle^{-s'})$. 
Proposition 5.2 shows that $f(P_2)-(P_2+i)^{-1}(P_1+i)f(P_1)$ has an asymptotic expansion in 
$\mathrm{Op}S(w^{-\infty}m^{-N-1}\langle x' \rangle^{-s'})$. 
We observe that 
\[f(P_2)-f(P_1)=\left(f(P_2)-(P_2+i)^{-1}(P_1+i)f(P_1)\right)+(P_2+i)^{-1}(V_1-V_2)f(P_1).\]
Theorem 5 and Proposition 5.2 show that 
the second term also has an asymptotic expansion in 
$\mathrm{Op}S(w^{-\infty}m^{-\infty}\langle x' \rangle^{-s'})$. 
Thus $f(P_2)-f(P_1)$ has an asymptotic expansion in 
$\mathrm{Op}S(w^{-\infty}m^{-N-1}\langle x' \rangle^{-s'})$. 
Thus $f(P_2)-f(P_1)$ has an asymptotic expansion in 
$\mathrm{Op}S(w^{-\infty}m^{-\infty}\langle x' \rangle^{-s'})
=\mathrm{Op}S(m^{-\infty}\langle x' \rangle^{-s'})$. 
Finally, we calculate the asymptotic expansion of $f(P_2)-f(P_1)$, whose existence has been proved now. 
Take $\chi \in C_c^{\infty}(\mathbb{R}^n)$ which is equal to $1$ on a large ball. 
We see from Theorem 5 that $(f(P_2)-f(P_1))\chi$ has an asymptotic expansion 
in $\mathrm{Op}S(m^{-\infty}\langle x' \rangle^{-s'})$ 
which coincides with the formal calculation. Since $\chi$ is arbitrary, 
we conclude that the asymptotic expansion of $f(P_2)-f(P_1)$ coincides with the formal one. 
\end{proof}

\appendix 
\section{Commutator calculation}
In this Appendix, we assume that $V \in C_b^{\infty}(\mathbb{R}^n ; \mathbb{R})$ and 
set $P=-\Delta + \beta  x_1+V(x)$. We denote Schwartz space and its dual by $\mathscr{S}$ and $\mathscr{S}'$.
To justify the commutator calculations in section 5, we prove the following;
\begin{prop}
For $\mathrm{Im}z \not =0$, $(P-z)^{-1}$ is continuous from $\mathscr{S}$ to $\mathscr{S}$. 
Thus, there is a unique continuous extension 
$(P-z)^{-1}: \mathscr{S}' \to \mathscr{S}'$ and this is the inverse of 
$P-z: \mathscr{S}' \to \mathscr{S}'$.  
In particular, $\mathrm{Ker}(P-z)=\{0\}$ on $\mathscr{S}'$. 
\end{prop}

This enables us to compute the commutator with the resolvent. 
\begin{corollary}
For any linear operator $T: \mathscr{S}' \to \mathscr{S}'$, 
$[T, (P-z)^{-1}]=-(P-z)^{-1}[T, P](P-z)^{-1}$ as an operator from 
$\mathscr{S}'$ to $\mathscr{S}'$. 
\end{corollary}
 
\begin{remark}
(1). We always have $(P-z)[T, (P-z)^{-1}]u=-[T, P](P-z)^{-1}u$ if $u, Tu \in L^2$. 
If we know that $[T, P](P-z)^{-1}u \in L^2$ and $[T, (P-z)^{-1}]u\in L^2$, 
we conclude that $[T, (P-z)^{-1}]u=-(P-z)^{-1}[T, P](P-z)^{-1}u$ since 
the domain of $P$ is $\{u\in L^2|Pu\in L^2\}$.   

(2). If we only know that $[T, P](P-z)^{-1}u \in L^2$, we cannot immediately conclude that 
$[T, (P-z)^{-1}]u\in L^2$ and $[T, (P-z)^{-1}]u=-(P-z)^{-1}[T, P](P-z)^{-1}u$. 
If we had a generalized eigenfunction $v\in \mathscr{S}'$ with  $(P-z)v=0$, 
there would be the possibility that $[T, (P-z)^{-1}]u=v-(P-z)^{-1}[T, P](P-z)^{-1}u \not \in L^2$. 
The above Proposition excludes this possibility.
\end{remark}
 
To apply the perturbation argument, we introduce the Banach space $Y^N=\bigcap_{k+s\le N}H^{k, s}$, 
where $H^{k, s}$ is the weighted Sobolev space 
\[H^{k, s}=\{u \in L^2 | \|u\|_{k,s}=\|\langle D \rangle^{k} \langle x \rangle^{s}u\|_{L^2}<\infty \}.\]
We only consider $k,s\in \mathbb{Z}_{\ge 0}$. 
The following proposition implies the Proposition A.1 
since $\mathscr{S}=\bigcap_{k ,s \ge 0}H^{k, s}$ including the topology.
\begin{prop}
For $\mathrm{Im}z \not =0$, $(P-z)^{-1}: Y^N \to Y^N$ is a bounded operator for any $N\ge 0$.
\end{prop} 
\begin{proof}
We first give a formal proof without justifying the commutator calculation.
Take $u \in Y^N$. Then for $k+s \le N$, 
\begin{align*}
&\|(P-z)^{-1}u\|_{k,s}=\|\langle D \rangle^{k} \langle x \rangle^{s}(P-z)^{-1}u\|_{L^2}\\
&\le \|(P-z)^{-1}[\langle D \rangle^{k} \langle x \rangle^{s}, P](P-z)^{-1}u\|_{L^2} 
	+\|(P-z)^{-1}\langle D \rangle^{k} \langle x \rangle^{s}u\|_{L^2} \\
&\le |\mathrm{Im}z|^{-1}\|[\langle D \rangle^{k} \langle x \rangle^{s}, P](P-z)^{-1}u\|_{L^2}
 +|\mathrm{Im}z|^{-1}\|u\|_{k,s}.
\end{align*}
Since $[\langle D \rangle^{k} \langle x \rangle^{s}, P]$ consists of the terms which can be estimated by 
$\langle D \rangle^{k-1} \langle x \rangle^{s}$ and $\langle D \rangle^{k+1} \langle x \rangle^{s-1}$, 
\[\|[\langle D \rangle^{k} \langle x \rangle^{s}, P](P-z)^{-1}u\|_{L^2} 
\lesssim \|(P-z)^{-1}u\|_{k-1, s}+\|(P-z)^{-1}u\|_{k+1, s-1}\]
(if k=0 or s=0, the first or the second term does not appear). 
Since one computation of the commutator adds $|\mathrm{Im}z|^{-1}$, 
the repetition of this procedure shows that 
\begin{equation}
\|(P-z)^{-1}\|_{Y^N \to Y^N}\le C_N |\mathrm{Im}z|^{-1} \max\{1, (1/|\mathrm{Im}z|)^{2N}\} \tag{A.1}
\end{equation}
if the above calculation is justified. 
We next give a rigorous proof.

We first assume that $V=0$. We set $P_0=-\Delta+\beta x_1$.
Then we have an explicit diagonalization 
$\mathcal{F}_{x'}\exp(-\frac{i}{3\beta}D_1^3)P_0\exp(\frac{i}{3\beta}D_1^3)
\mathcal{F}_{x'}^{-1}=|\xi'|^2+\beta x_1$, 
where $\mathcal{F}_{x'}$ is the Fourier transform with respect to $x'$. 
Since  
$\mathcal{F}_{x'}\exp(-\frac{i}{3\beta}D_1^3)$ and $(|\xi'|^2+\beta x_1-z)^{-1}$ preserve $\mathscr{S}$, 
we conclude that $(P_0-z)^{-1}$ preserves $\mathscr{S}$. 
Thus Proposition A.1 and Corollary A.1 are true for $V=0$.
Then the above calculation is justified and the estimate (A.1) is true for $P_0-z$.

We next assume that $V\in C_b^{\infty}(\mathbb{R}^n; \mathbb{R})$ and fix $N\ge 0$. 
We note that $V$ is a bounded operator from $Y^N$ to $Y^N$. 
This and the estimate (A.1) for $P_0$ imply that there exists $\rho_0>0$ 
such that $\|(P_0-z)^{-1}V\|_{Y^N \to Y^N}<1$ for $|\mathrm{Im}z|> \rho_0$. 
Thus the Neumann series argument shows that 
 $(P-z)^{-1}=(1+(P_0-z)^{-1}V)^{-1}(P_0-z)^{-1}$ is bounded from $Y^N$ to $Y^N$ for $|\mathrm{Im}z|> \rho_0$. 
Then the above calculation is justified by Remark A.1.(1) 
and the a priori estimate (A.1) (rather than the estimate from the Neumann series argument)
is true for $P-z$ with $|\mathrm{Im}z|> \rho_0$. 

We next weaken the assumption that $|\mathrm{Im}z|> \rho_0$.
Take $z_0$ with $|\mathrm{Im}z_0|> \rho_0$. 
If $|z-z_0| C_N |\mathrm{Im}z_0|^{-1} \max\{1, ({1}/{|\mathrm{Im}z_0|})^{2N}\}<1$, 
the estimate (A.1) for $P-z_0$ and the Neumann series argument show that 
$(P-z)^{-1}=(1+(z_0-z)(P-z_0)^{-1})^{-1}(P-z_0)^{-1}$ is bounded from $Y^N$ to $Y^N$. 
Thus the above calculation is justified by Remark A.1.(1) and the estimate (A.1) is true for $P-z$. 
Since $|\mathrm{Im}z_0|> \rho_0$ is arbitrary, (A.1) is true for $P-z$ with $|\mathrm{Im}z|> \rho_1$, where 
$\rho_1=\rho_0-(C_N \rho_0^{-1} \max\{1, ({1}/{\rho_0})^{2N}\})^{-1}$. 

The repetition of this argument shows that 
the estimate (A.1) is true for $P-z$ with $|\mathrm{Im}z|> \rho_j$, where 
$\rho_j=\rho_{j-1}-(C_N \rho_{j-1}^{-1} \max\{1, ({1}/{\rho_{j-1}})^{2N}\})^{-1}$. 
We may assume that $C_N>1$ and thus $\rho_j>0$. 
Since $\rho_0 > \rho_1 >\rho_2 > \cdots >0$, 
there exists $\rho_{\infty}=\lim_{j \to \infty}\rho_j$. 
To finish the proof, it is enough to show that $\rho_{\infty}=0$. 
Assume on the contrary that $\rho_{\infty}>0$. 
Then $\rho_{j-1}-\rho_j=(C_N \rho_{j-1}^{-1} \max\{1, ({1}/{\rho_{j-1}})^{2N}\})^{-1}>
(C_N \rho_{\infty}^{-1} \max\{1, ({1}/{\rho_{\infty}})^{2N}\})^{-1}$ for any $j$. 
Thus $\lim_{j \to \infty}\rho_j=-\infty$, which is a contradiction.
\end{proof}

\begin{remark}
All the results in this Appendix are true for $\beta=0$. 
The free diagonalization is of course the Fourier transform. 
If we replace $|\mathrm{Im}z|$ by $\mathrm{dist}(z, \sigma(P))$ in the proof, 
the results in this case are also true for any $z$ in the resolvent set $\mathbb{C}\setminus \sigma(P)$.
\end{remark}

\section*{Acknowledgement}
The author is grateful to his advisor Shu Nakamura for discussions and the encouragement. 
The author is also grateful to the anonymous referee for valuable suggestions to 
improve the manuscript. 
The author is under the support of the FMSP program at the 
Graduate School of Mathematical Sciences, the University of Tokyo.

Graduate School of Mathematical Sciences, University of Tokyo, 
3-8-1, \\Komaba, Meguro-ku, Tokyo 153-8914, Japan

E-mail address: kameoka@ms.u-tokyo.ac.jp

\end{document}